\documentclass{amsproc}

\usepackage[T2A]{fontenc}
\usepackage[cp1251]{inputenc}
\usepackage[english]{babel}

\usepackage{amsfonts}
\usepackage{amsthm}
\usepackage{amssymb}
\usepackage{amsmath}
\usepackage{amsopn}
\usepackage{bm}
\usepackage{array}
\usepackage{color}
\usepackage{mathdots}
\usepackage{graphicx}
\usepackage{longtable}

\newtheorem{prop}{Proposition}
\newtheorem{theo}{Theorem}
\newtheorem{lemma}{Lemma}

\newtheorem{conj}{Conjecture}

\theoremstyle{definition}

\newtheorem{rem}{Remark}

\usepackage{cite}
\allowdisplaybreaks[4]

\DeclareMathOperator{\Ad}{Ad}

\DeclareMathOperator{\tr}{tr}

\DeclareMathOperator{\cn}{cn}

\DeclareMathOperator{\al}{Al}

\newcommand{\rmx}{\mathrm{x}}
\newcommand{\rmt}{\mathrm{t}}
\newcommand{\rmr}{\mathrm{r}}
\newcommand{\rmh}{\mathrm{h}}
\newcommand{\rmc}{\mathrm{c}}

\newcommand{\Lmatr}{\mathsf{L}}

\newcommand{\sfx}{\mathsf{x}}
\newcommand{\sft}{\mathtt{t}}

\newcommand{\rmd}{\mathrm{d}}
\DeclareMathOperator{\modR}{mod}

\newcommand{\M}{\mathcal{M}}

\newcommand{\rmb}{\mathrm{b}}

\DeclareMathOperator{\Complex}{\mathbb{C}}
\DeclareMathOperator{\Real}{\mathbb{R}}

\DeclareMathOperator{\Integer}{\mathbb{Z}}

\DeclareMathOperator{\Jac}{\mathrm{Jac}}

\DeclareMathOperator{\const}{const}

\DeclareMathOperator{\card}{card}

\DeclareMathOperator{\ReN}{\mathrm{Re}}
\DeclareMathOperator{\ImN}{\mathrm{Im}}

\allowdisplaybreaks

\title[Exact quasi-periodic solutions to the mKdV equation]
{Exact quasi-periodic solutions to the MKdV equation}
\author{Julia Bernatska}
\address{University of Connecticut, Department of Mathematics}
\email{julia.bernatska@uconn.edu}

\begin{document}
 
\maketitle

\begin{abstract}
In the present paper, a hierarchy of the mKdV equation 
is integrated by the methods of algebraic geometry. The mKdV hierarchy in question
arises on coadjoint orbits in the loop algebra of $\mathfrak{sl}(2)$, and
employs a family of hyperelliptic curves as spectral curves.
A generic form of the finite-gap solution in any genus is obtained in terms of
the $\wp$-functions, which generalize the Weierstrass $\wp$-function.
Reality conditions for quasi-periodic wave solutions are completely specified.  
The obtained solutions are illustrated by plots in small genera.
\end{abstract}

% 35B15  Almost and pseudo-almost periodic solutions to PDEs
% 35B30 Dependence of solutions to PDEs on initial and/or boundary data and/or on parameters of PDEs [See also 37Cxx
% 35Q53 KdV equations (Korteweg-de Vries equations) {For dynamical systems and ergodic theory, see 37K10}

% 37J35 Completely integrable finite-dimensional Hamiltonian systems, integration methods, integrability tests
% 37J37 Relations of finite-dimensional Hamiltonian and Lagrangian systems with Lie algebras and other algebraic structures
% 37J38 Relations of finite-dimensional Hamiltonian and Lagrangian systems with algebraic geometry, complex analysis, special functions
% 37K10 Completely integrable infinite-dimensional Hamiltonian and Lagrangian systems, integration methods, inte-
%    grability tests, integrable hierarchies (KdV, KP, Toda, etc.)
%  37K15 Inverse spectral and scattering methods for infinite-dimensional Hamiltonian and Lagrangian systems
% 37K20 Relations of infinite-dimensional Hamiltonian and Lagrangian dynamical systems with algebraic geometry,
%  complex analysis, and special functions [See also 14H70]

%======================================
\section{Introduction}

The mKdV equation 
 \begin{equation}\label{MKdVEq}
w_{\rmt} + w_{\rmx \rmx \rmx}  + 6 \varsigma w^2 w_{\rmx} =0
\end{equation} 
is known in two forms: defocusing (or negative) with $\varsigma\,{=}\,{-}1$,
and  focusing (or positive) with $\varsigma=1$.
Significant efforts have been dedicated to solving the both equations,
since they  describe multiple physical phenomena.

In \cite{Miu1968}, it was discovered that every solution of the defocusing mKdV equation 
gives rise  to a solution of the KdV equation
\begin{equation}\label{KdVEq}
\upsilon_{\rmt} + \upsilon_{\rmx \rmx \rmx}  - 6 \upsilon \upsilon_{\rmx} =0,
\end{equation} 
by the Miura transformation
\begin{equation}\label{MiuraTrans}
\upsilon = w_{\rmx}  + w^2.
\end{equation}

Multi-soliton solutions to the focusing mKdV equation are obtained
by the Hirota bilinear method in \cite{Hirota1972}, and by the inverse scattering method in \cite{Wad1973}.
The B\"{a}cklund  transformation method is applied to this equation in \cite{Ono1976}.

In \cite{GesSvi1995},
the concept of relative scattering is applied to constructing 
soliton solutions relative to general  mKdV background solutions,
and to the particular case of
$N$-solitons relative to quasi-periodic finite-gap backgrounds. 
No explicit solutions to the defocusing mKdV equation are employed,
instead the corresponding solutions to the KdV equation, 
namely the $N$-soliton solution, and the finite gap solution through the theta function,
are suggested to use, since solutions of these two equations are connected by the 
Miura transformation.

Breather lattice solutions to the focusing mKdV are suggested in \cite{KKS2003}, 
and also to the defocusing mKdV in \cite{KKH2004}.
In \cite{ZFML2009}, 26 more kinds of breather lattice solutions
to the defocusing mKdV are presented. Also
travelling wave solutions written explicitly through the Jacobi elliptic functions is suggested in \cite{Fu2004El}, 
and through the trigonometric and hyperbolic functions in \cite{Fu2004Trig}.

The  stability analysis of periodic stationary wave solutions to the both mKdV
equations is presented in \cite{KapDec2015}, and of periodic travelling wave solutions
in \cite{DecNiv2010}. The $\cn$-periodic solution to the focusing mKdV equation appears to be 
modulationally unstable, and on the $\cn$-periodic background
serves as a model of the classical rogue wave, as shown in \cite{ChenPen2018}.
A solution that describes the interaction of a dark soliton and a travelling periodic wave
is derived in \cite{MucPel2024}.

An exact hyperelliptic solution to the focusing mKdV equation is suggested in \cite{Mat2002}.
This solution is expressed through the $\al$-functions introduced by Weierstrass, and
obtained by the Miura transformation from the $\wp_{1,1}$-solution to the KdV equation.
In \cite{Mat2022} this solution is applied to modelling configurations of the supercoiled DNA.
The latter is modelled as an elastic rod, with the thermal effect taken into account.
Then, the isometric and iso-energy deformation 
is determined so that the curvature obeys the focusing mKdV equation.
Exact treatment in genera one and two is presented,
and solutions are evaluated numerically. In \cite{Mat2024} real hyperelliptic 
solutions to the focusing mKdV equation are constructed in genus three.
As a part of investigation, conditions on branch points
of the spectral curve are suggested.

In the present paper, quasi-periodic wave solutions to the both mKdV equations 
are obtained. These are finite-gap solutions, which arise in the mKdV hierarchy of hamiltonian systems
constructed on coadjoint orbits in the loop algebra of $\mathfrak{sl}(2,\Real)$
for the defocusing mKdV equation, and of $\mathfrak{su}(2)$ for the focusing one.
Solutions in terms of the multiply periodic $\wp$-functions are obtained by algebraic integration.
These solutions work for an arbitrary genus of the  spectral curve associated with a hamiltonian system,
and represent non-linear waves with two independent parameters $\rmx$, and $\rmt$.
%The obtained finite-gap solutions cover a wide variety of .
In particular, in genus one these solutions turn into travelling wave solutions.
%the corresponding expressions through the Jacobi elliptic functions are also presented.
Multi-soliton solutions can be obtained by degeneration of the spectral curve of high genus to 
an elementary one. A proper degeneration of the spectral curve also leads to (multi-)soliton
solutions on a (quasi-)periodic background, which is a matter of further investigation.

In fact, the mKdV equation arises as a dynamical equation on the Jacobian variety of
the spectral curve in question, and has the form of an identity for the $\wp$-functions.
In other words, this dynamical equation holds for every point of the Jacobian variety,
except those which correspond to special divisors. 

The problem of reality conditions is solved completely for quasi-periodic wave solutions.
That is, the obtained solutions are real-valued, and bounded functions of $\rmx$ and $\rmt$.
The conditions on the spectral curve are also revealed, since not every real curve
can serve as spectral within the mKdV hierarchy, 
as well as the KdV, and sine(sinh)-Gordon hierarchies, see
\cite{BerKdV2024,BerSinG2024}.

The paper is organized as follows. 
In Section 2 the mKdV hierarchy is briefly described, and the both forms of the mKdV equations
are derived from  hamiltonian flows. Separation of variables is briefly presented in Section 3,
where also the equations of motion for variables of separation are obtained. And so
the connection between $\rmx$ and $\rmt$, on the one hand, and the argument of $\wp$-functions,
on the other hand, is revealed. In Section 4 the problem of uniformization of a family of hyperelliptic curves,
which serve as spectral within the mKdV hierarchy, is described, and a solution 
to the Jacobi inversion problem in terms of the $\wp$-functions is recalled.
Section 5 is devoted to algebraic integration, and all dynamical variables are expressed in terms of the $\wp$-functions.
In Section~6 a quasi-periodic solution, which works for every finite-gap system of the mKdV hierarchy,
is presented, and reality conditions are specified. The both focusing, and defocusing
cases are analysed, since they have different  reality conditions.
In Section 7 a new quasi-periodic solution to the KdV equation is obtained by the Muira transformation.
%Finally, in Section 8 non-linear wave solutions in genera $1$, $2$, and $3$ are presented,
%including plots.
%\textcolor{red}{In Appendix the addition law.}

%======================================
\section{Integrable systems on coadjoint orbits of a loop group}\label{s:KdVHier}
The mKdV equation arises within the hierarchy of integrable hamiltonian systems 
on coadjoint orbits in the loop algebra of $\mathfrak{sl}(2,\Real)$ (defocusing mKdV),
and $\mathfrak{su}(2)$ (focisung mKdV). Here the hierarchy is built by the orbit method
in the form presented  in \cite{Holod84}. In more detail this approach is exposed 
in \cite{BerKdV2024}, and \cite{BerSinG2024}, where the hierarchies of the KdV,
 and sine(sinh)-Gordon equations are presented, respectively. 
 
 %-----------------------------------------------
\subsection{Phase space of mKdV hierarchy}
The three hierarchies are constructed on the loop Lie algebra of  
$\widetilde{\mathfrak{g}} = \mathfrak{g} \otimes
\Complex[z,z^{-1}]$, $\mathfrak{g}$ is one of the real forms of $\mathfrak{sl}(2,\Complex)$, with principal grading.
According to the Adler---Kostant---Symes scheme, see \cite{AdlMoer80},  
$\widetilde{\mathfrak{g}}$ is decomposed into two parts:
\begin{gather*}
\widetilde{\mathfrak{g}} = \widetilde{\mathfrak{g}}_+ \oplus \widetilde{\mathfrak{g}}_-,\qquad
\widetilde{\mathfrak{g}}_+ = \textstyle \oplus \sum_{\ell \geqslant 0} \mathfrak{g}_{\ell},\quad
\widetilde{\mathfrak{g}}_- = \oplus \sum_{\ell < 0} \mathfrak{g}_{\ell},
\end{gather*}
where $\mathfrak{g}_{\ell}$ denotes the eigenspace of the grading operator of degree $\ell$.

The dual algebra $\widetilde{\mathfrak{g}}^\ast$ for the mKdV hierarchy
is introduced by the same bilinear form as for the KdV hierarchy, see \cite[Sect.\,3]{BerKdV2024}.
This induces the decomposition 
\begin{gather*}
\widetilde{\mathfrak{g}}^\ast = \widetilde{\mathfrak{g}}_+^{\ast} \oplus \widetilde{\mathfrak{g}}_-^{\ast},\qquad
\widetilde{\mathfrak{g}}_+^\ast = \textstyle \oplus \sum_{\ell \leqslant -1} \mathfrak{g}_{\ell},\quad
\widetilde{\mathfrak{g}}_-^\ast = \widetilde{\mathfrak{g}}_+ \oplus \mathfrak{g}_{-1}.
\end{gather*}

The common  finite-gap subspace $\M_N$ for the hierarchies is
\begin{gather}\label{MKdVPhSp}
\M_N = \bigg\{\Lmatr = \begin{pmatrix} \alpha(z) & \beta(z) \\ \gamma(z) & -\alpha(z)
\end{pmatrix} \mid \begin{array}{l} \alpha(z) = \sum_{m=0}^{N} \alpha_{2m} z^m,\\
\beta(z) = \sum_{m=0}^{N} \beta_{2m-1} z^{m-1},\\
\gamma(z) = \sum_{m=0}^{N} \gamma_{2m-1} z^{m},
\end{array}\ 
\begin{array}{l}  \alpha_{2N} = 0,\\ 
\beta_{2N-1} = \gamma_{2N-1} \\
\ \ = \rmb = \const.
\end{array} \bigg\}.
\end{gather}
In the case of the mKdV hierarchy, similar to the KdV case, it arises as
$\M_N = \widetilde{\mathfrak{g}}^\ast_- / \big( \sum_{\ell > 2N }  \mathfrak{g}_\ell \big)$.
The coordinates $\{\beta_{2m -1}, \gamma_{2m - 1}, \alpha_{2m}\}_{m=0}^{N-1}$ of $\M_N$, 
of number $3N$, serve as dynamic variables of  $N$-gap hamiltonian systems.

The subspace $\M_N$ is equipped with the Lie-Poisson bracket,
which introduces the symplectic structure, the same as for the KdV hierarchy,
\begin{gather}\label{MKdVPoiBra}
\begin{split}
&\{\beta_{2m-1}, \alpha_{2n}\} = \beta_{2(n+m)+1},\\
&\{\gamma_{2m-1}, \alpha_{2n}\} = - \gamma_{2(n+m)+1},\\
&\{\beta_{2m-1}, \gamma_{2n-1}\} = - 2 \alpha_{2(m+n)},\quad 0 \leqslant m+n \leqslant N.
\end{split}
\end{gather}

The action of $\widetilde{G}_- = \exp(\widetilde{\mathfrak{g}}_-)$ splits $\M_N$ into orbits
$$\mathcal{O} = \{\Lmatr = \Ad^\ast_{g} \Lmatr^{\text{in}} \mid g\in \widetilde{G}_-\},\qquad  \Lmatr^{\text{in}} \in \M_N.$$
An initial point $ \Lmatr^{\text{in}}$ is spanned by $H^{\ast}_{2m}$, $m = 0$, \ldots, $N$,
and belongs to the Weyl chamber of $\widetilde{G}_-$. 
Each orbit $\mathcal{O}$, $\dim \mathcal{O} = 2N$,
 serves as the \emph{phase space} of an $N$-gap hamiltonian system
 of the mKdV hierarchy.

%-----------------------------------------------
\subsection{Integrals of motion}\label{ss:IntMot}
Integrals of motion $h_n$ arise as coefficients of the polynomial $H$
invariant under the action of $\widetilde{G}_{-}$,  obtained by
\begin{equation}\label{InvF}
\begin{split}
H(z) &= \tfrac{1}{2} \tr \Lmatr^2 = \alpha(z)^2 + \beta(z) \gamma(z)\\
&= h_{2N-1} z^{2N-1} + \cdots + h_1 z + h_0 + h_{-1} z^{-1},
\end{split}
\end{equation}
where $z$ serves as a spectral parameter. Thus,
\begin{subequations}\label{hExprs}
\begin{align}
& h_{2N-1} = \rmb^2 \equiv \const,\\
& h_{2N-2} = \alpha_{2N-2}^2 + \rmb (\beta_{2N-3} + \gamma_{2N-3}), \label{MKdVConstr}\\
&\dots \notag \\
& h_{0} = \alpha_0^2 + \beta_1 \gamma_{-1} + \beta_{-1} \gamma_{1},\\
& h_{-1} = \beta_{-1} \gamma_{-1}.
\end{align}
\end{subequations}

With respect to the symplectic structure \eqref{MKdVPoiBra},
$h_{-1}$, $h_0$, \ldots, $h_{N-2}$ give rise to non-trivial hamiltonian flows, we call them \emph{hamiltonians}. 
The other integrals of motion $h_{N-1}$, $h_N$, \ldots, $h_{2N-2}$ annihilate 
the Poisson bracket: $\{h_n, \mathcal{F}\} = 0$
for any  $\mathcal{F} \in \mathcal{C}^1 (\M_N)$. The latter produce  the \emph{constraints}
\begin{gather}\label{MKdVConstr}
h_{N-1} = \rmr_{N-1},\quad h_{N} = \rmr_{N},\quad  \ldots,\quad  h_{2N-2} = \rmr_{2N-2}.
\end{gather}
The constants $\rmr_{N-1}$, $\rmr_N$, \ldots, $\rmr_{2N-2}$
fix an orbit~$\mathcal{O}$ of dimension $2N$ in $3N$-dimensional subspace $\M_N$,
and so fix the phase space of a hamiltonian system.

\begin{rem}
Note, in the sine(sinh)-Gordon hierarchy $h_{N-1}$, $h_N$, \ldots, $h_{2N-2}$ serve as hamiltonians,
and $h_{-1}=\rmr_{-1}$, $h_0=\rmr_0$, \ldots, $h_{N-2}=\rmr_{N-2}$ fix an orbit.
\end{rem}

Integrals of motion $h_n$ are real,
if $\mathfrak{g}$ is one of the real forms of $\mathfrak{sl}(2,\Complex)$:
\begin{itemize}
\item $\mathfrak{g} = \mathfrak{sl}(2,\Real)$, all coordinates 
$\beta_{2m -1}$, $\gamma_{2m - 1}$, $\alpha_{2m}$ are real, 
and  $\rmb \in \Real$;
\item $\mathfrak{g} = \imath \mathfrak{sl}(2,\Real)$ with 
$\rmb$, $\beta_{2m -1}$, $\gamma_{2m - 1} \,{\in}\, \imath \Real$, and $\alpha_{2m} \,{\in}\, \Real$, $m=0$, \ldots $N-1$;
\item $\mathfrak{g} = \mathfrak{su}(2)$, which implies $\alpha_{2m} \,{=}\, \imath a_{2m}$, $a_{2m} \in \Real$,
and $\beta_{2m-1} \,{=}\, {-} \bar{\gamma}_{2m-1}$ ($\bar{\gamma}$ denotes the complex conjugate of~$\gamma$),  
and $\rmb = \imath b$,  $b \in \Real$.
\item $\mathfrak{g} = \imath \mathfrak{su}(2)$, the dynamic variables are subject to 
$\alpha_{2m} \,{=}\, \imath a_{2m}$, $a_{2m} \,{\in}\, \Real$,
and $\beta_{2m-1} \,{=}\, \bar{\gamma}_{2m-1}$, and $\rmb \in \Real$.
\end{itemize}

%-----------------------------------------------
\subsection{Focusing and defocusing mKdV equations}\label{ss:SineSinhParam}
We consider the two hamiltonian flows: 
\begin{gather}\label{TwoFlowsEq}
\frac{\partial \psi_{a,\ell}}{\partial \sfx} = \{h_{N-2}, \psi_{a,\ell}\},\qquad
\frac{\partial \psi_{a,\ell}}{\partial \sft} = \{h_{N-3}, \psi_{a,\ell}\}.
\end{gather}  
The stationary flow generated by $h_{N-2}$ has the form
\begin{subequations}
\begin{align}
&\frac{\partial \beta_{-1}}{\partial \sfx} =  2 \alpha_{2N-2} \beta_{-1}, \label{BetaN1StEq}\\
&\frac{\partial \gamma_{-1}}{\partial \sfx} = - 2 \alpha_{2N-2} \gamma_{-1}, \label{GammaN1StEq}\\
&\frac{\partial \beta_{2m-1}}{\partial \sfx} = 2 \alpha_{2N-2} \beta_{2m-1} - 2 \rmb \alpha_{2m-2} ,
\label{BetaStEq} \\
&\frac{\partial \gamma_{2m-1}}{\partial \sfx} = - 2 \alpha_{2N-2} \gamma_{2m-1} + 2 \rmb \alpha_{2m-2}, 
\quad m=1,\dots, N-1, \label{GammaStEq} \\
&\frac{\partial \alpha_{2m}}{\partial \sfx} =  \rmb (\gamma_{2m-1} - \beta_{2m-1}),\quad m=0,\dots, N-1.
\label{AlphaStEq}
\end{align}
\end{subequations}
From the evolutionary flow generated by $h_{N-3}$ we use the equation
\begin{align}\label{Alpha2N2EvEq}
&\frac{\partial \alpha_{2N-2}}{\partial \sft} = \rmb \big( \gamma_{2N-5} - \beta_{2N-5} \big).
\end{align}

The mKdV equation arises from the equality
\begin{gather}\label{MKdVEqZeroCurv}
\frac{\partial \alpha_{2N-2}}{\partial \sft} = \frac{\partial \alpha_{2N-4}}{\partial \sfx},
\end{gather}
which follows from \eqref{Alpha2N2EvEq}, and \eqref{AlphaStEq}  with $m=N-2$.
From \eqref{BetaStEq} and \eqref{GammaStEq} with $m=N-1$ we find
\begin{gather}\label{Alpha2N4}
\alpha_{2N-4} = \frac{1}{4\rmb} \bigg(
\frac{\partial}{\partial \sfx}  \big(\gamma_{2N-3} - \beta_{2N-3}\big) 
+ 2 \alpha_{2N-2} (\gamma_{2N-3} + \beta_{2N-3} ) \bigg).
\end{gather}
Then, \eqref{MKdVConstr}, and \eqref{AlphaStEq}  with $m=N-1$ imply, respectively,
\begin{subequations}\label{MKdVAuxEqs}
\begin{align}
&\gamma_{2N-3} + \beta_{2N-3} = \rmb^{-1} \big(\rmr_{2N-2} - \alpha_{2N-2}^2\big),\\
&\gamma_{2N-3} - \beta_{2N-3} = \frac{1}{\rmb} \frac{\partial \alpha_{2N-2}}{\partial \sfx}.
\end{align}
\end{subequations}
By substituting \eqref{MKdVAuxEqs} into \eqref{Alpha2N4},
and then substituting the obtained expression for $\alpha_{2N-4}$ into \eqref{MKdVEqZeroCurv},
we obtain the extended mKdV equation for  $\alpha_{2N-2} \equiv \alpha(\sfx,\sft)$
\begin{equation}\label{MKdVEqPre} 
\frac{\partial \alpha}{\partial \sft} = \frac{1}{4\rmb^2}
\frac{\partial}{\partial \sfx} \bigg(
\frac{\partial^2 \alpha}{\partial \sfx^2}  
+ 2 \alpha (\rmr_{2N-2} - \alpha^2)
\bigg).
\end{equation}
After the change of variables 
\begin{equation}\label{xtNatural}
(\sfx, \sft) \mapsto (\rmx,\rmt),\qquad
\sfx = \rmx + 2 \rmr_{2N-2} \rmt,\quad 
\sft = -4 \rmb^2 \rmt,
\end{equation}
we come to the standard mKdV equation \eqref{MKdVEq}.

In the cases of $\mathfrak{g} \,{=}\, \mathfrak{sl}(2,\Real)$,
and $\mathfrak{g} \,{=}\, \imath \mathfrak{sl}(2,\Real)$ the dynamic variable $\alpha_{2N-2}$
is real-valued.
Then $w(\rmx,\rmt) \,{=}\, \alpha (\rmx + 2 \rmr_{2N-2} \rmt,-4 \rmb^2 \rmt)$
obeys  the \emph{defocusing mKdV equation}  \eqref{MKdVEq}, $\varsigma=1$.

If $\mathfrak{g} \,{=}\, \mathfrak{su}(2)$,
or $\mathfrak{g} \,{=}\, \imath \mathfrak{su}(2)$, we assign  
$w(\rmx,\rmt) \,{=}\, {-} \imath \alpha (\rmx + 2 \rmr_{2N-2} \rmt,-4 \rmb^2 \rmt)$,
which is  real-valued, and
satisfies the \emph{focusing mKdV equation}  \eqref{MKdVEq}, $\varsigma=-1$.

\begin{rem}
The mKdV equation arises when $N \geqslant 2$.
In the case of $N=1$,  there exists one hamiltonian $h_{-1} = \beta_{-1}  \gamma_{-1}$, which
produces the stationary flow 
\begin{align*}
&\frac{\partial \beta_{-1}}{\partial \sfx} = 2 \alpha_{0} \beta_{-1},&
&\frac{\partial \gamma_{-1}}{\partial \sfx} = - 2 \alpha_{0} \gamma_{-1}, &
&\frac{\partial \alpha_{0}}{\partial \sfx} =  \rmb (\gamma_{-1} - \beta_{-1}).
\end{align*}
From these equations, and the constraint $\rmr_0 = \alpha_0^2 + \rmb (\gamma_{-1} + \beta_{-1})$,
we find the equation for  $\alpha_0$
\begin{gather}\label{MKdVG1EqPre}
 \frac{\partial^2 \alpha_0}{\partial \sfx^2} = - 2 \alpha_{0} \big(\rmr_0 - \alpha_0^2 \big).
\end{gather}
After the change of variable $\sfx \mapsto \rmx$ 
such that  $\sfx = \rmx + 2 \rmr_{0} \rmt$,  we come to the 
defocusing mKdV equation for $w(\rmx,\rmt) \,{=}\, \alpha_0(\rmx + 2 \rmr_{0} \rmt)$,
if $\mathfrak{g} \,{=}\, \mathfrak{sl}(2,\Real)$,
or $\mathfrak{g} \,{=}\, \imath \mathfrak{sl}(2,\Real)$, and 
to the focusing mKdV equation for $w(\rmx,\rmt) \,{=}\, {-} \imath \alpha_0(\rmx + 2 \rmr_{0} \rmt)$,
if $\mathfrak{g} \,{=}\, \mathfrak{su}(2)$,
or $\mathfrak{g} \,{=}\, \imath \mathfrak{su}(2)$.
\end{rem}

\begin{rem}
In what follows, we work with the parameters $\sfx$, $\sft$ of the stationary and evolutionary flows, respectively,
which serve as natural parameters of the proposed mKdV hierarchy. In section~\ref{s:NLW},
finite-gap solutions and their graphical representation are given with respect to
the parameters $\rmx$, $\rmt$ of the mKdV equation.
\end{rem}

%-----------------------------------------------
\subsection{Zero curvature representation}
The system of dynamic equations \eqref{TwoFlowsEq} admits  the matrix form
\begin{gather}\label{PsiMatrEqs} 
\frac{\rmd \Psi}{\partial \sfx} = [\Psi,   \mathrm{A}_\text{st}], \qquad
\frac{\rmd \Psi}{\partial \sft} = [\Psi,   \mathrm{A}_\text{ev}], \\
\begin{split}
& \mathrm{A}_\text{st} = - \begin{pmatrix} 
\alpha_{2N-2} & \rmb \\  
\rmb z  & -\alpha_{2N-2} 
 \end{pmatrix},\\
& \mathrm{A}_\text{ev} =  - z \begin{pmatrix} 
\alpha_{2N-2} & \rmb \\  
\rmb z  & -\alpha_{2N-2} 
 \end{pmatrix}
 - \begin{pmatrix} 
\alpha_{2N-4} & \beta_{2N-3} \\  
\gamma_{2N-3} z  & -\alpha_{2N-4} 
 \end{pmatrix}.
 \end{split} \notag
\end{gather}
Then \eqref{MKdVEqZeroCurv} arises on the diagonal of
the matrix zero curvature representation
\begin{gather*}
\frac{\rmd  \mathrm{A}_\text{st}}{\partial \sft} - \frac{\rmd  \mathrm{A}_\text{ev}}{\partial \sfx} 
= [ \mathrm{A}_\text{st},  \mathrm{A}_\text{ev}].
\end{gather*}

%======================================
\section{Separation of variables}\label{s:SoV}

%-----------------------------------------------
\subsection{Spectral curve}
The spectral curve of the mKdV hierarchy is obtained from 
the characteristic polynomial of $\Psi$, namely $\det \big(\Psi(z) \,{-}\, (w/z)\big) \,{=}\, 0$,
which leads to the equation
\begin{equation*}
- w^2 + z^2 H(z) = 0,
\end{equation*}
where  $H$ is defined by \eqref{InvF}. This equation
 represents a genus $N$ hyperelliptic curve 
\begin{multline}\label{CurveGN}
 0 = F(z,w) \equiv - w^2 + \rmb^2 z^{2N+1} + \rmr_{2N-2} z^{2N} + \cdots + \rmr_{N-1} z^{N+1} \\
 + \rmh_{N-2} z^{N} + \dots + \rmh_0 z^2 + \rmh_{-1} z.
\end{multline}
All parameters of the curve are integrals of motion: the hamiltonians $h_{-1}$, \ldots, $h_{N-2}$, 
the constraints $\rmr_{N-1}$, \dots, $\rmr_{2N-2}$, and $h_{2N-1} = \rmb^2 = \const$.

%-----------------------------------------------
\subsection{Canonical coordinates}

As shown in \cite{KK1976} for the sine-Gordon hierarchy, and in \cite{BerHol07}
for the hierarchies with hyperelliptic spectral curves, variables of separation 
are given by a positive non-special\footnote{A non-special divisor on a hyperelliptic curve of genus $g$
is a degree $g$ positive divisor which contains no pairs of points in involution.}  
divisor of degree equal to the genus $N$ of the spectral curve.
In fact, pairs of coordinates of $N$ points from the support of such a divisor serve as quasi-canonical variables, 
and so lead to separation of variables.

Following the procedure proposed in \cite{BerHol07}, we eliminate
the dynamic variables  $\beta_{2m-1}$, $m=0$, \ldots, $N-1$, with the help of the constraints \eqref{MKdVConstr},
and obtain

\begin{theo}\label{T:SoVdiv}
Let the phase space $\mathcal{O}$, $\dim \mathcal{O} \,{=}\, 2N$,
of a hamiltonian system of the mKdV hierarchy
be parametrized by the dynamic variables
$\{\gamma_{2m-1}$, $\alpha_{2m} \mid m\,{=}\,0, \dots, N\,{-}\,1\}$ which satisfy the constraints \eqref{MKdVConstr}.
Then the points $\{(z_k,w_k)\}_{k=1}^{N}$ which form the divisor of zeros of the system
\begin{equation}\label{PointsDef1}
\gamma(z_k) = 0,\qquad w_k = \epsilon z_k \alpha(z_k),  \quad k =1,\dots, N,
\end{equation}
$\epsilon^2 \,{=}\, 1$, belong to the spectral curve \eqref{CurveGN}, and form a non-special divisor.
\end{theo}
\begin{theo}\label{T:QCanonVars}
The coordinates $\{(z_k,w_k)\}_{k=1}^{N}$ defined in Theorem\;\ref{T:SoVdiv}
have the following properties:
\begin{enumerate}
\renewcommand{\labelenumi}{\arabic{enumi})}
\item a pair $(z_k, w_k)$ is quasi-canonically conjugate with respect to the  Lie-Poisson
bracket \eqref{MKdVPoiBra}:
\begin{equation}\label{CanonCoord}
\{z_k, z_j\} = 0,\qquad  \{z_k, w_j\} = \epsilon\, z_k\, \delta_{k,j},\qquad
\{w_k, w_j\} = 0.
\end{equation}
\item the canonical $1$-form is
\begin{equation}\label{Liouv1Form}
 \epsilon \sum_{k=1}^N  z_k w_k \rmd z_k.
\end{equation}
\end{enumerate}
\end{theo}
Proofs of the theorems are similar to the ones given in \cite[Sect.\,4]{BerKdV2024}.

In what follows we assign $\epsilon = 1$.

%-----------------------------------------------
\subsection{Equation of motion for variables of separation}
From \eqref{PsiMatrEqs} we find
\begin{gather}\label{DgammaDx}
\begin{split}
&\frac{\rmd}{\rmd \sfx} \gamma(z) = - 2 \alpha_{2N-2} \gamma(z) + 2 \rmb z \alpha(z),\\
&\frac{\rmd}{\rmd \sft} \gamma(z) = - 2 (\alpha_{2N-2}z + \alpha_{2N-4}) \gamma(z) 
+ 2 (\rmb z^2 + \gamma_{2N-3} z) \alpha(z).
\end{split}
\end{gather}
Note that zeros of $\gamma(z)$ are  functions of $\sfx$ and $\sft$, that is
 $\gamma(z) = \rmb \prod_{k=1}^N (z-z_k(\sfx,\sft))$,
and all dynamic variables are functions of $\sfx$ and $\sft$ as well. 
From \eqref{DgammaDx}, taking into account 
\eqref{PointsDef1}, we find
\begin{gather}\label{DzDxEqs}
\frac{\rmd z_k}{\rmd \sfx}
= \frac{- 2 w_k}{\prod_{j\neq k}^N (z_k - z_j)},\qquad
\frac{\rmd z_k}{\rmd \sft}
 = \frac{2 w_k \sum_{j\neq k} z_j}{\prod_{j\neq k} (z_k - z_j) }, \quad k=1,\, \dots,\, N.
\end{gather}

Let $D = \sum_{k=1}^N (z_k,w_k)$ be a non-special divisor on the spectral curve
related to the dynamic variables $\gamma_{2m-1}$, $\alpha_{2m}$, $m=0$, \ldots, $N\,{-}\,1$,
by \eqref{PointsDef1}.
Then
\begin{equation*}
u = \mathcal{A}(D) =
\sum_{k=1}^N \int_{\infty}^{(z_k,w_k)} 
\small \begin{pmatrix} z^{N-1} \\ \vdots \\ z \\ 1 \end{pmatrix} 
\frac{\rmb \rmd z}{-2w}.
\end{equation*}
 Applying \eqref{DzDxEqs}, we find
\begin{subequations}
\begin{align*}
&\frac{\rmd u_{2n-1}}{\rmd \sfx} = \sum_{k=1}^N \frac{\rmb z_k^{N-n}}{-2w_k} \frac{\rmd z_k}{\rmd \sfx}
= \sum_{k=1}^N \frac{\rmb z_k^{N-n}}{\prod_{j\neq k}^N (z_k - z_j)} =  \rmb \delta_{n,1}, \\
& \frac{\rmd u_{2n-1}}{\rmd \sft} = \sum_{k=1}^N \frac{\rmb z_k^{N-n}}{-2w_k} \frac{\rmd z_k}{\rmd \sft}
= - \sum_{k=1}^N  \frac{\rmb z_k^{N-n}\sum_{j\neq k} z_j}{\prod_{j\neq k} (z_k - z_j) } = \rmb \delta_{n,2},
\end{align*}
\end{subequations}
and so  obtain
\begin{gather}\label{uInxt}
\begin{split}
u_{1} &= \rmb \sfx + C_1,\qquad
u_{3} = \rmb \sft + C_{3},  \\
u_{2n-1} &= \rmb \rmc_{2n-1}  + C_{2n-1}, = \const,\quad n=3, \ldots, N.
\end{split}
\end{gather}
Here  $\rmc_{2n-1}$, $n=3, \ldots, N$, are arbitrary real constants, and
 $\bm{C} = (C_1$, $C_3$, \ldots, $C_{2N-1})^t$ is a fixed constant vector subject to the reality conditions.

%===============================
\section{Uniformization of the spectral curve}
The spectral curve \eqref{CurveGN} is reduced to the canonical  form  $\mathcal{V}$ of genus $g=N$
\begin{equation}\label{V22g1Eq}
f(x,y) \equiv -y^2 + x^{2 g+1} + \sum_{i=0}^{2g} \lambda_{2i+2} x^{2g-i} = 0,
\end{equation}
by  applying the transformation
\begin{equation}\label{SectrCToCanonC}
\begin{split}
z &\mapsto x,\qquad 
w \mapsto  y = w/\rmb,\qquad F(x, \rmb y) = \rmb^2 f(x,y),\\
&\rmh_{n}\; (\text{or } \rmr_n)= \rmb^2 \lambda_{4g-2-2n},\ \ n=-1,\,\dots,\, 2N-2.
\end{split}
\end{equation}
Note, that one of the branch points of \eqref{CurveGN} is fixed at $0$,
and so the corresponding canonical  form has $\lambda_{4g+2}=0$.

%-------------------------------------------------------
\subsection{Family of hyperelliptic curves}\label{ss:HypC}
Let finite branch points of \eqref{V22g1Eq} be 
 $(e_i,0)$, and denoted by $e_i$, $i=1$, \ldots, $2g+1$ for brevity. 
 One more branch point $e_0$  is located at infinity, and serves as the basepoint.
We assume that all branch points are distinct, and so the curve is not degenerate, that is, the genus equals $g$.

Finite branch points are enumerated in the ascending order of their real and imaginary parts.
The Riemann surface of $\mathcal{V}$ is constructed by means of the
 monodromy path through all branch points, according to the order; the path starts at infinity and ends at infinity,
see the orange line on fig.~\ref{cyclesOdd}. 
For more details on constructing the Riemann surfaces of hyperelliptic curves 
suitable for computation see \cite[Sect.\,3]{BerCompWP2024}.

A homology basis is defined after H.\,Baker \cite[p.\,297]{bakerAF}.
Cuts are made between points $e_{2k-1}$ and $e_{2k}$ with $k$ from $1$ to $g$. 
One more cut starts at  $e_{2g+1}$ and ends at infinity.
Canonical homology cycles are defined as follows.
Each $\mathfrak{a}_k$-cycle, $k=1$, \ldots $g$, encircles the cut $(e_{2k-1},e_{2k})$ counter-clockwise,
and each $\mathfrak{b}_k$-cycle emerges from the cut $(e_{2g+1},\infty)$ 
and enters the cut $(e_{2k-1},e_{2k})$, see  fig.\;\ref{cyclesOdd}. 

\begin{figure}[h]
\begin{flushleft}
 1a. All real branch points 
\end{flushleft} 
\includegraphics[width=0.67\textwidth]{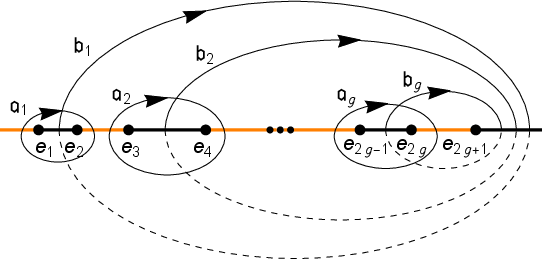} 
\begin{flushleft}
1b. Branch points with $g$ complex conjugate pairs
\end{flushleft} 
\includegraphics[width=0.67\textwidth]{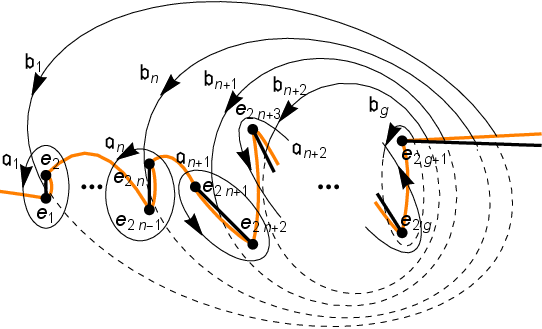} \\
\caption{Cuts and cycles on a hyperelliptic curve.} \label{cyclesOdd}
\end{figure}

We are interested in two cases: (i) when all finite branch points are real, see fig.~\ref{cyclesOdd}a,
and (ii) $2g$ branch points form $g$ complex conjugate pairs
and the remaining finite branch point is zero, $e_{2n+1}=0$, see fig.~\ref{cyclesOdd}b.
Fig.~\ref{cyclesOdd}a represents a spectral curve of the KdV hierarchy, and
fig.~\ref{cyclesOdd}b a spectral curve of  the sine-Gordon hierarchy, 
as shown in \cite[Sect.\,6]{BerKdV2024}, and \cite[Sect.\,6]{BerSinG2024},
respectively.

Let first kind differentials $\rmd u = (\rmd u_1,\rmd u_3, \dots, \rmd u_{2g-1})^t$ 
and second kind differentials $\rmd r = (\rmd r_1,\rmd r_3, \dots, \rmd r_{2g-1})^t$ 
on $\mathcal{V}$ form a system of
 associated differentials, see \cite[\S\,138]{bakerAF}, and be
defined as in \cite[\textit{Ex.\,i}, p.\,195]{bakerAF}.
Actually, 
\begin{subequations}\label{K1K2DifsGen}
\begin{align}
& \rmd u_{2n-1} =  \frac{x^{g-n} \rmd x}{\partial_y f(x,y)},\quad n=1,\dots, g,\label{K1DifsGen} \\
& \rmd r_{2n-1} =  \frac{ \rmd x}{\partial_y f(x,y)} 
 \sum_{j=1}^{2n-1} (2n-j) \lambda_{2j-2} x^{g+n-j},\  \lambda_0=1,\quad n=1,\dots, g. \label{K2DifsGen}
\end{align}
\end{subequations}
Indices of $\rmd u_{2n-1}$ display the orders of zeros at infinity, 
and  indices of $\rmd r_{2n-1}$ display the orders of poles, which are located at infinity.
Note the differentials are not normalized.

First kind periods along the canonical cycles $\mathfrak{a}_k$,  $\mathfrak{b}_k$, $k=1$, \ldots, $g$, 
 are defined as follows
\begin{gather}\label{FKD}
  \omega_k = \oint_{\mathfrak{a}_k} \rmd u,\qquad\qquad
  \omega'_k = \oint_{\mathfrak{b}_k} \rmd u.
\end{gather}
The vectors $\omega_k$, $\omega'_k$ 
form first kind period matrices $\omega$, $\omega'$, respectively. They also  generate 
the period lattice $\{\omega, \omega'\}$, with respect to which the Jacobian variety of $\mathcal{V}$
is defined: $\Jac(\mathcal{V}) \,{=}\, \Complex^g/\{\omega, \omega'\}$.
Let $u=(u_1$, $u_3$, \ldots, $u_{2g-1})^t$ be not coordinates (normalized) of $\Jac(\mathcal{V})$.

The Abel map is defined by
\begin{gather*}%\label{AbelM}
 \mathcal{A}(P) = \int_{\infty}^P \rmd u,\qquad P=(x,y)\in \mathcal{V},
\end{gather*}
and on  a positive divisor $D = \sum_{k =1}^n (x_k,y_k)$  by $\mathcal{A}(D) =  \sum_{k =1}^n  \mathcal{A}(P_k)$.
The Abel map is one-to-one on the $g$-th symmetric power of the curve.

Similarly to \eqref{FKD}, second kind period matrices $\eta$, $\eta'$ are composed of the columns
\begin{gather}\label{SKD}
  \eta_k = \oint_{\mathfrak{a}_k} \rmd r,\qquad\qquad
  \eta'_k = \oint_{\mathfrak{b}_k} \rmd r.
\end{gather}

%----------------------------------------
\subsection{Entire functions}
The Riemann theta function  is defined by 
\begin{gather*}%\label{ThetaDef}
 \theta(v;\tau) = \sum_{n\in \Integer^g} \exp \big(\imath \pi n^t \tau n + 2\imath \pi n^t v\big),
\end{gather*}
where $v = \omega^{-1}u$ are normalized coordinates,
and the Riemann period matrix $\tau$ is defined by $\tau = \omega^{-1}\omega'$,
and  belongs to the Siegel upper half-space.
The theta function with characteristic $[\varepsilon]$ is defined by
\begin{equation*}%\label{ThetaDefChar}
 \theta[\varepsilon](v;\tau) = \exp\big(\imath \pi (\varepsilon'{}^t/2) \tau (\varepsilon'/2)
 + 2\imath \pi  (v+\varepsilon/2)^t \varepsilon'/2\big) \theta(v+\varepsilon/2 + \tau \varepsilon'/2;\tau),
\end{equation*}
where $[\varepsilon]= (\varepsilon', \varepsilon)^t$
 is a $2\times g$ matrix, all components of $\varepsilon$, and $\varepsilon'$
are real values within the interval $[0,2)$. Modulo $2$ addition
is defined on characteristics. 

Every point $u$ within the fundamental domain of  $\Jac(\mathcal{V})$ 
can be represented by its characteristic $[\varepsilon]$ as follows
\begin{equation*}
u =  \tfrac{1}{2}  \omega \varepsilon +  \tfrac{1}{2}  \omega' \varepsilon'.
\end{equation*}
Characteristics with values $0$ and $1$ correspond to
half-periods, which are Abel images of divisors composed of branch points.
Such a characteristic $[\varepsilon]$ is odd whenever $\varepsilon^t \varepsilon'  = 1$ ($\modR 2$), 
and even whenever $\varepsilon^t \varepsilon' = 0$ ($\modR 2$). The theta function with characteristic
has the same parity as its characteristic.

The modular invariant entire function on $\Complex^g \,{\supset}\, \Jac(\mathcal{V})$  is called the sigma function, 
which we define after \cite[Eq.(2.3)]{belHKF}:
\begin{equation}\label{SigmaThetaRel}
\sigma(u) = C \exp\big({-}\tfrac{1}{2} u^t \varkappa u\big) \theta[K](\omega^{-1} u;  \omega^{-1} \omega'),
\end{equation}
where $[K]$ is the characteristic of the vector $K$ of Riemann constants, and 
$\varkappa = \eta \omega^{-1}$ is a symmetric matrix.

In what follows we use the multiply periodic $\wp$-functions, see \cite[]{bakerAF},
also known as the Klinian $\wp$-functions, see \cite{belHKF},
\begin{gather}\label{WPdef}
\wp_{i,j}(u) = -\frac{\partial^2 \log \sigma(u) }{\partial u_i \partial u_j },\qquad
\wp_{i,j,k}(u) = -\frac{\partial^3 \log \sigma(u) }{\partial u_i \partial u_j \partial u_k},
\end{gather}
which are defined on $\Jac(\mathcal{V}) \backslash \Sigma$, 
where $\Sigma = \{u \mid \sigma(u)=0\}$ denotes the theta divisor, see \cite[p.\,38]{DN1982},
 in not normalized coordinates.

%-------------------------------------
\subsection{Jacobi inversion problem}
A solution of the Jacobi inversion problem on a hyperelliptic curve is proposed in  \cite[Art.\;216]{bakerAF},
see also \cite[Theorem 2.2]{belHKF}.
Let $u = \mathcal{A}(D )$ be the Abel image of  a non-special positive divisor  
$D \in \mathcal{V}^g$. Then $D$ is uniquely defined by the system of equations 
\begin{subequations}\label{EnC22g1}
\begin{align}
&\mathcal{R}_{2g}(x;u) \equiv x^{g} -  \sum_{i=1}^{g} x^{g-i}  \wp_{1,2i-1}(u) = 0, \label{R2g}\\ 
&\mathcal{R}_{2g+1}(x,y;u) \equiv 2 y + \sum_{i=1}^{g} x^{g-i}  \wp_{1,1,2i-1}(u) = 0. \label{R2g1}
\end{align}
\end{subequations}

The functions $\wp_{1,2i-1}$, $\wp_{1,1,2i-1}$, $i=1$, \ldots, $g$,
uniformize the hyperelliptic curve \eqref{V22g1Eq},
see \cite[Sect.\,4]{BerCompWP2024},
and serve as basis $\wp$-functions 
in the abelian function field associated with $\mathcal{V}$, see \cite[Sect.\,3.2]{BerWPFF2025}.

%======================================
\section{Algebraic integration}\label{s:AGI} 

Separation of variables provides a divisor defined by \eqref{PointsDef1}
on the spectral curve \eqref{CurveGN}. On the other hand,
such a divisor is defined by a solution of the Jacobi inversion problem, which is given by
\eqref{EnC22g1} on the canonical curve \eqref{V22g1Eq}. 

We will work with $\wp$-functions associated with $\mathcal{V}$. 
All dynamic variables of a $N$-gap hamiltonian system are
expressible in terms of the basis functions $\wp_{1,2i-1}(u)$, $\wp_{1,1,2i-1}(u)$, $i=1$, \ldots $N$. 

\begin{lemma}
The dynamic variables $\{\gamma_{2i-1}$, $\alpha_{2i} \mid i\,{=}\,0, \dots, N\,{-}\,1\}$,
which describe a $2N$-dimensional phase space in the mKdV hierarchy, 
are expressed through the basis $\wp$-functions as follows
\begin{subequations}\label{WPAlphaGamma}
\begin{align}
&\gamma_{2(N-i)-1} = - \rmb \wp_{1,2i-1}(u), \quad i=1,\dots, N; \label{WPGamma}\\
&\alpha_{2N-2} = - \frac{\rmb \wp_{1,1,2N-1}(u)}{2 \wp_{1,2N-1}(u)}, \\
&\alpha_{2(N-i)} = - \frac{\rmb}{2} \Big(\wp_{1,1,2i-3}(u) -  
\frac{\wp_{1,2i-3}(u) \wp_{1,1,2N-1}(u)}{\wp_{1,2N-1}(u)} \Big), 
\ \  i=2,\dots, N.  \label{WPAlpha}
\end{align}
\end{subequations}
\end{lemma}
\begin{proof}
Taking into account the transformation \eqref{SectrCToCanonC}, we find
how $\rmd u$ and $\rmd r$  are expressed in terms of the coordinates of \eqref{CurveGN},
namely
\begin{subequations}
\begin{align}
& \rmd u_{2n-1} =  \frac{\rmb z^{N-n} \rmd z}{\partial_w F(z,w)},\quad n=1,\dots, N,\label{K1Difs} \\
& \rmd r_{2n-1} =  \frac{ \rmd z}{\rmb \partial_w F(z,w)} \sum_{j=1}^{2n-1} (2n-j) \rmh_{2N-j} z^{N+n-j}, 
\quad n=1,\dots, N, \label{K2Difs}
\end{align}
\end{subequations}
where the notation $\rmh_n$ is used for $\rmh_n$ and $\rmr_n$, and $\rmh_{2N-1} = \rmb^2$.
Then the Abel pre-image $D = \sum_{k=1}^N (z_k,w_k)$ of any $u\in \Jac(\mathcal{V}) \backslash \Sigma$
is obtained from  \eqref{EnC22g1}.

On the other hand, the $N$ values $z_k$ are zeros of the polynomial
$\gamma(z)$, and the $N$ values $w_k$ are obtained from $w_k = z_k \alpha(z_k)$, cf.\,\eqref{PointsDef1}. 
\end{proof}

Expressions for $\beta_{2n-1}$, $n=0$, \ldots, $N-1$ 
are obtained from the constraints \eqref{MKdVConstr},
in particular (the argument $u$ is omitted for brevity)
\begin{subequations}\label{WPBeta}
\begin{align}
&\beta_{2N-3} =  \rmb \bigg( {-} \frac{\wp_{1,1,2N-1}^2}{4 \wp_{1,2N-1}^2} + \wp_{1,1} + \lambda_2 \bigg),\\
&\beta_{2N-5} = \rmb \bigg(\frac{\wp_{1,1,2N-1}^2}{4 \wp_{1,2N-1}^2} \wp_{1,1}
- \frac{\wp_{1,1,2N-1}}{2 \wp_{1,2N-1}}  \wp_{1,1,1}
+\wp_{1,3} + \wp_{1,1}^2 + \lambda_2 \wp_{1,1} + \lambda_4\bigg), \label{Beta2N5} \\
&\beta_{2N-7} = \rmb \bigg(\frac{\wp_{1,1,2N-1}^2}{4 \wp_{1,2N-1}^2} \wp_{1,3}
- \frac{\wp_{1,1,2N-1}}{2 \wp_{1,2N-1}}  \wp_{1,1,3} - \tfrac{1}{4} \wp_{1,1,1}^2 + \wp_{1,5} \\
&\qquad\qquad + (2\wp_{1,1}  + \lambda_2) \wp_{1,3} + \wp_{1,1}^3 
+ \lambda_2 \wp_{1,1}^2 + \lambda_4 \wp_{1,1} + \lambda_6 \bigg), \notag 
\quad \text{etc.}
\end{align}
\end{subequations}

The expressions for $\alpha_{2m}$, $\beta_{2m-1}$, $m=0$, \ldots, $N-1$,
can be simplified by introducing $\wp$-functions not from the basis.

%---------------------------------------
\begin{theo}
The dynamic variables $\{\gamma_{2m-1}$, $\beta_{2m-1}$, $\alpha_{2m} \mid m\,{=}\,0, \dots, N\,{-}\,1\}$
of the $N$-gap subspace $\mathcal{M}_N$ are expressed as follows
\begin{subequations}\label{WPAllM}
\begin{align}
&\gamma_{2(N-i)-1} = - \rmb \wp_{1,2i-1}(u), \ \  i=1,\dots, N; \\
&\alpha_{2(N-i)} = - \frac{\rmb \wp_{1,2i-1,2N-1}(u)}{2 \wp_{1,2N-1}(u)},
\ \  i=1,\dots, N;  \label{WPAlphaM}\\
\begin{split}
&\beta_{-1} = \frac{- \rmh_{-1}}{\rmb \wp_{1,2N-1}(u)},\quad
\beta_{2(N-i)-1} = \rmb \frac{ \wp_{2i+1,2N-1}(u)}{\wp_{1,2N-1}(u)},
\ \   i=1,\dots, N-1.  \label{WPBetaM}
\end{split}
\end{align}
\end{subequations}
\end{theo}
\begin{proof}

By substitution of \eqref{WPAlphaGamma} into equations of the stationary flow,
identities for $\wp_{1,1,1,2i-1}$, $i=1$, \ldots, $g$, are obtained from \eqref{AlphaStEq}, 
in particular 
\begin{subequations}\label{4IndWP}
\begin{align}
&\rho_{2N+2} =  - \wp_{1,1,1,2N-1} + 4 \wp_{1,1} \wp_{1,2N-1} + 2\lambda_2 \wp_{1,2N-1}
+ \frac{1}{2} \frac{\wp_{1,1,2N-1}^2}{\wp_{1,2N-1}}, \label{WP1112n1}\\
&\rho_{4}  = -\wp_{1,1,1,1} + 6 \wp_{1,1}^2 + 4 \wp_{1,3} + 4 \lambda_2 \wp_{1,1} + 2 \lambda_4. \label{WP1111}
\end{align}
\end{subequations}
Equations \eqref{GammaN1StEq}, \eqref{GammaStEq} vanish,
and \eqref{BetaN1StEq}, \eqref{BetaStEq} are equivalent to the identities for $\wp_{1,1,1,2i-1}$.

Next, the evolutionary flow has the form
\begin{subequations}
\begin{align}
&\frac{\partial \alpha_{0}}{\partial \sft} =  \beta_{2N-3} \gamma_{-1} - \gamma_{2N-3}\beta_{-1}, \qquad
\frac{\partial \beta_{-1}}{\partial \sft} =  2 \alpha_{2N-4} \beta_{-1}, \notag \\
&\frac{\partial \gamma_{-1}}{\partial \sft} = - 2 \alpha_{2N-4} \gamma_{-1}, \label{GammaN1EvEq}\\
&\frac{\partial \beta_{2m-1}}{\partial \sft} = 
2 (\alpha_{2N-4} \beta_{2m-1} + 2 \alpha_{2N-2} \beta_{2m-3}) \notag\\ %\label{BetaEvEq}  
&\qquad\qquad\ \ - 2 (\rmb \alpha_{2m-4} + \beta_{2N-3} \alpha_{2m-2}), \notag \\
&\frac{\partial \gamma_{2m-1}}{\partial \sft} = 
- 2 (\alpha_{2N-4} \gamma_{2m-1} + \alpha_{2N-2} \gamma_{2m-3}) \label{GammaEvEq}  \\
&\qquad\qquad\ \ + 2 (\rmb \alpha_{2m-4} + \gamma_{2N-3} \alpha_{2m-2}), \notag\\
&\frac{\partial \alpha_{2m}}{\partial \sft} = \rmb \big( \gamma_{2m-3} - \beta_{2m-3} \big) \notag \\ %\label{AlphaEvEq} 
&\qquad\qquad + \beta_{2N-3} \gamma_{2m-1} - \gamma_{2N-3}\beta_{2m-1},
\quad m=1,\dots, N-1.  \notag
\end{align}
\end{subequations}
Substituting \eqref{WPAlphaGamma} into \eqref{GammaN1EvEq}, 
\eqref{GammaEvEq}, we find the identities
\begin{gather}\label{3IndWP}
\begin{split}
&\wp_{1,3,2i-1} = \wp_{1,1,2i+1} 
+ \wp_{1,2i-1} \wp_{1,1,1} - \wp_{1,1} \wp_{1,1,2i-1} ,\quad i=2, \dots, N-1,\\
&\wp_{1,3,2N-1} = \wp_{1,2N-1} \wp_{1,1,1} - \wp_{1,1} \wp_{1,1,2N-1}.
\end{split}
\end{gather}
All other equations of the evolutionary flow are reduced to \eqref{3IndWP}
and the identities for $\wp_{1,1,1,2i-1}$ obtained from the stationary flow.

From the flows generated by other hamiltonians $h_{N-4}$, \ldots, $h_{-1}$
we find
\begin{gather}\label{3IndWPExt}
\wp_{1,2i-1,2N-1} = \wp_{1,2N-1} \wp_{1,1,2i-3} - \wp_{1,2i-3} \wp_{1,1,2N-1},\quad i=3, \dots, N.
\end{gather}

With the help of \eqref{3IndWP}, and \eqref{3IndWPExt}, the expressions \eqref{WPAlpha} 
for $\alpha_{2m}$, $m=0$, \ldots, $N-2$, 
are simplified as follows
\begin{align}\label{WPAlphaM}
&\alpha_{2(N-i)} = - \frac{\rmb \wp_{1,2i-1,2N-1}(u)}{2 \wp_{1,2N-1}(u)},
\quad i=1,\dots, N.
\end{align}

Expressions for $\beta_{2m-1}$, $m=0$, \ldots, $N-1$, are simplified by introducing  
functions $\wp_{2i-1,2N-1}$ through the fundamental cubic relations,
see \cite[Theorem 3.2]{belHKF},
\begin{multline}\label{WP3IndSqRel}
\wp_{1,1,2i-1} \wp_{1,1,2j-1} = 4(\wp_{1,1}+\lambda_2) \wp_{1,2i-1} \wp_{1,2j-1} \\
+ 4 (\wp_{1,2i-1} \wp_{1,2j+1} + \wp_{1,2i+1} \wp_{1,2j-1}) + 4 \wp_{2i+1,2j+1} \\
- 2 (\wp_{3,2i-1} \wp_{1,2j-1} + \wp_{1,2i-1} \wp_{3,2j-1}) - 2 (\wp_{2i-1,2j+3} + \wp_{2i+3,2j-1}) \\
+ \lambda_4 (\wp_{1,2j-1} \delta_{1,i}  + \wp_{1,2i-1}\delta_{1,j} )
+ 4 \lambda_{2+4i} \delta_{i,j} + 2 (\lambda_{4i} \delta_{i-1,j} + \lambda_{4j} \delta_{i,j-1}).
\end{multline}
The fundamental cubic relations enable to eliminate $\wp_{1,1,2i-1}\wp_{1,1,2j-1}$, $i$, $j\,{=}\,1$, \ldots, $g$,
which leads to \eqref{WPBetaM}.
\end{proof}

\begin{rem}
Expression from \eqref{hExprs} for the hamiltonians, written
in terms of the basis $\wp$-functions,  produce algebraic equations of $\Jac(\mathcal{V})$.
\end{rem}
% $\lambda_{2N+2}$, \ldots, $\lambda_{4N}$

\begin{theo}\label{T:MKdVESol}
The $N$-gap solution to the extended mKdV equation \eqref{MKdVEqPre} is 
\begin{equation}\label{MKdVExtSol}
\alpha(\sfx, \sft)  = - \frac{\rmb \wp_{1,1,2N-1}\big(\rmb (\sfx, \sft, \rmc_5,\dots, \rmc_{2N-1})^t + \bm{C}\big)}
{2 \wp_{1,2N-1}\big(\rmb (\sfx, \sft, \rmc_5,\dots, \rmc_{2N-1})^t + \bm{C} \big)},\quad \sfx, \sft \in \Real,
\end{equation}
where $\rmc_{2i-1}$, $i=3$, \ldots, $N$, are arbitrary real constants, and
$\bm{C}$ is a fixed constant vector subject to the reality conditions.
\end{theo}
This statement follows immediately from \eqref{WPAlphaM}, and \eqref{uInxt}.

\begin{rem}
Note that \eqref{MKdVExtSol} can be written in the form
\begin{equation}
\alpha(\sfx, \sft)  = - \tfrac{1}{2}  \partial_x \log \wp_{1,2N-1}
\big(\rmb (\sfx, \sft, \rmc_5,\dots, \rmc_{2N-1})^t + \bm{C} \big).
\end{equation}
\end{rem}

\begin{rem}
By substituting \eqref{MKdVExtSol}, \eqref{Beta2N5}, and \eqref{WPGamma} at $k=2$
into \eqref{Alpha2N2EvEq}, the following identity is obtained
\begin{multline}\label{MKdVWPEq}
 \frac{\wp_{1,1,3,2N-1}}{2 \wp_{1,2N-1}} - \frac{\wp_{1,1,2N-1}\wp_{1,3,2N-1}}{2 \wp_{1,2N-1}^2} \\
= \frac{\wp_{1,1} \wp_{1,1,2N-1}^2}{4 \wp_{1,2N-1}^2}
-  \frac{\wp_{1,1,1} \wp_{1,1,2N-1}^2}{2 \wp_{1,2N-1}}
- 2 \wp_{1,3} - \wp_{1,1}^2 - \lambda_2 \wp_{1,1} - \lambda_4,
\end{multline}
which serves as the dynamic equation on $\Jac(\mathcal{V}) \backslash \Sigma$ 
which represents the mKdV equation, if $\lambda_2=0$.
Indeed, using \eqref{WP1111} we eliminate $\lambda_4$ from \eqref{MKdVWPEq}. 
Then, with the help of \eqref{WP1112n1},
we express $\wp_{1,1}$, $\wp_{1,1,1}$, and $\wp_{1,1,1,1}$ in terms of
 $\partial_{u_1}^i \wp_{1,2N-1}$, $i=0$, $1$, \ldots, $4$,
 and come to the mKdV equation  for  $-\rmb \wp_{1,1,2N-1}(u)/\wp_{1,2N-1}(u)$.
 \end{rem}

%======================================
\section{Finite-gap solutions to the mKdV equation}\label{s:NgapMKdVSol}
We are looking for quasi-periodic solutions, which means we require $\alpha(\sfx, \sft)$
defined by \eqref{MKdVExtSol} to be real-valued and bounded.
Below, based on the results of \cite{BerKdV2024} and \cite{BerSinG2024}, 
\emph{reality conditions} are specified, 
including conditions on the spectral curve,
and on the constant vector $\bm{C}$.
The both cases of focusing, and defocusing mKdV are considered.

%-----------------------------------------------
\subsection{Reality conditions}\label{s:RealCond}
The reality conditions on hyperelliptic curves are closely related to the half-period lattice.
Let  $\mathfrak{J}^{\ReN} \,{=}\, \{\Omega \,{+}\, s \mid s\,{\in}\, \Real^g\}$  be affine subspaces
parallel to the real axes, where $\Omega$ runs over all  half-periods.
There exist $2^g$ non-congruent subspaces of this type.
Let  $\mathfrak{J}^{\ImN} \,{=}\, \{\Omega \,{+}\, \imath s \mid s\,{\in}\, \Real^g\}$  be affine subspaces
parallel to the imaginary axes, where $\Omega$ runs over all  half-periods.
There also exist $2^g$ non-congruent subspaces $\mathfrak{J}^{\ImN}$.

\begin{prop}\label{P:AllRealBPs}
Let all finite branch points of $\mathcal{V}$ be real, then the 
period lattice is rectangular, and  $\wp_{i,j}$ are real-valued on all subspaces $\mathfrak{J}^{\ReN}$, 
and $\mathfrak{J}^{\ImN}$;
$\wp_{i,j,k}$ are real-valued  on $\mathfrak{J}^{\ReN}$, and purely imaginary on $\mathfrak{J}^{\ImN}$.
\end{prop}
\begin{proof}
According to  \cite[Proposition \,1]{BerKdV2024}, we have
 $\omega_k \,{\in}\,\Real^g$, and $\omega'_k \,{\in}\, \imath \Real^g$ if 
the homology basis is introduced as in Subsection\;\ref{ss:HypC}. 
Thus, each pair $\omega_k$, $\omega'_k$ generates a rectangular sublattice, and 
the whole period lattice is rectangular. 
As proven in \cite[Propositions\,4, and 5]{BerKdV2024}, 
$\wp_{i,j}$ are real-valued on all subspaces $\mathfrak{J}^{\ReN}$, and $\mathfrak{J}^{\ImN}$;
$\wp_{i,j,k}$ are real-valued  on $\mathfrak{J}^{\ReN}$, and purely imaginary on $\mathfrak{J}^{\ImN}$.
\end{proof}

Let some branch points  of $\mathcal{V}$ form complex conjugate pairs, then
the period lattice is composed of rectangular and rhombic sublattices.
Each complex conjugate pair brings to a rhombic sublattice.
The maximal number of complex conjugate pairs is $g$, then
all  sublattices are rhombic, see \cite[Theorem\,3]{BerSinG2024},  and we say that the whole period lattice is rhombic.

In the presence of complex conjugate pairs among finite branch points of $\mathcal{V}$,
$\wp_{i,j}$ are real-valued only on the two subspaces: 
$\mathfrak{J}_0^{\ReN} \,{=}\, \{s \mid s\,{\in}\, \Real^g\}$, and 
$\mathfrak{J}_0^{\ImN} \,{=}\, \{\imath s \mid s\,{\in}\, \Real^g\}$;
$\wp_{i,j,k}(u)$ are real-valued on $\mathfrak{J}_0^{\ReN}$, and
purely imaginary on $\mathfrak{J}_0^{\ImN}$.

Another thing which should be taken into consideration is singularities of $\wp$-functions. 
The singularities are
located at odd non-singular, and all singular half-periods, where the $\sigma$-function vanishes.
As proven in \cite[Propositions\,2, and 3]{BerKdV2024}, and \cite[Theorem\,4]{BerSinG2024},
in $\Jac(\mathcal{V})$ there exist only two affine subspaces free of zeros of the $\sigma$-function:
$\mathfrak{J}_K^{\ReN} \,{=}\, \{u[K] \,{+}\, s \mid s\,{\in}\, \Real^g\}$, and 
$\mathfrak{J}_K^{\ImN} \,{=}\, \{u[K] \,{+}\, \imath s \mid s\,{\in}\, \Real^g\}$.
This implies, that bounded solutions to the integrable hierarchy are obtained with $\bm{C} = u[K]$.

 \begin{theo}\label{T:Cshift}
The constant $\bm{C}$ from Theorem~\ref{T:MKdVESol} is defined by
\begin{equation}\label{CDef}
\bm{C} = \sum_{j=0}^{[(g-1)/2]} \tfrac{1}{2}  \omega_{g-2j}  + \sum_{k=1}^g \tfrac{1}{2}  \omega'_k
=  \mathcal{A} \Big(\sum_{i=1}^g e_{2i} \Big),
\end{equation}
and corresponds to the vector of Riemann constants $u[K]$. 
\end{theo}

%-----------------------------------------------
\subsection{Defocusing mKdV hierarchy}
This hierarchy arises in the cases 
\begin{itemize}
\item  $\mathfrak{g} = \mathfrak{sl}(2,\Real)$, then 
\begin{equation}\label{SL2Cond}
\alpha_{2m}, \ \beta_{2m -1},\  \gamma_{2m - 1} \in  \Real,\  \
m=0,\, \dots,\, N-1,\quad \rmb \in \Real;
\end{equation}

\item  $\mathfrak{g} = \imath \mathfrak{sl}(2,\Real)$, then
\begin{equation}\label{iSL2Cond}
\alpha_{2m}\in \Real, \quad \beta_{2m -1},\  \gamma_{2m - 1} \in  \imath \Real,\  \
m=0,\, \dots,\, N-1,\quad \rmb \in \imath \Real.
\end{equation}
\end{itemize}
Taking into account Proposition\;\ref{P:AllRealBPs}, we see that all expressions given by
\eqref{WPAllM} comply with \eqref{SL2Cond} on $u\in \mathfrak{J}^{\ReN}$,
and with \eqref{iSL2Cond} on $u\in \mathfrak{J}^{\ImN}$, if and only if all
finite branch points of $\mathcal{V}$ are real.

\begin{theo}
The defocusing mKdV equation \eqref{MKdVEq}, $\varsigma=1$, 
arises on hyperelliptic curves which possess a branch point
at infinity, a branch point at the origin, and all other branch points are real.  
The real-valued, and bounded finite-gap solution
is given by $w(\rmx,\rmt) = \alpha(\rmx,\rmt)$, where 
\begin{equation}\label{MKdVSol}
 \alpha(\rmx,\rmt)  = - \frac{\rmb 
\wp_{1,1,2N-1} \big( \rmb (\rmx + 2 \rmr_{2N-2} \rmt,-4 \rmb^2 \rmt, \rmc_5,\dots, \rmc_{2N-1})^t + u[K]\big)}
{2 \wp_{1,2N-1}\big( \rmb (\rmx + 2 \rmr_{2N-2} \rmt,-4 \rmb^2 \rmt, \rmc_5,\dots, \rmc_{2N-1})^t + u[K]\big)},
\end{equation}
$\rmb \in \Real$, or $\rmb \in \imath \Real$.
\end{theo}
\begin{proof}
The statement follows from Theorem~\ref{T:MKdVESol} after applying the transformation
\eqref{xtNatural}, and taking into account Proposition\;\ref{P:AllRealBPs},
and Theorem\;\ref{T:Cshift}.
\end{proof}

%-----------------------------------------------
\subsection{Focusing mKdV hierarchy}
The hierarchy arises in the cases
\begin{itemize}
\item  $\mathfrak{g} = \mathfrak{su}(2)$, then
\begin{equation}\label{SU2Cond}
\alpha_{2m} \in \imath \Real,\quad 
\beta_{2m -1} =- \bar{\gamma}_{2m - 1},  \  m=0,\, \ldots,\, N-1,\quad \rmb \in \imath \Real;
\end{equation}

\item  $\mathfrak{g} = \imath \mathfrak{su}(2)$, then
\begin{equation}\label{iSU2Cond}
\alpha_{2m} \in \imath \Real,\quad  
\beta_{2m -1} =\bar{\gamma}_{2m - 1}, \ m=0,\, \ldots,\, N-1, \quad \rmb \in \Real.
\end{equation}
\end{itemize}

\begin{conj}\label{C:CCBPs}
All expressions given by
\eqref{WPAllM} comply with \eqref{SU2Cond} on $u\in \mathfrak{J}^{\ImN}_K$,
and with \eqref{iSU2Cond} on $u\in \mathfrak{J}^{\ReN}_K$, if and only if 
branch points of $\mathcal{V}$ are $0$, $\infty$, and $g$ complex conjugate pairs.
\end{conj}
\begin{proof}
The proof uses uniformization of the curve in question, given by \eqref{WPAllM}.

Case $\mathfrak{g} \,{=}\, \mathfrak{su}(2)$. Then  $\rmb \,{=}\, \imath b$, $b \,{\in}\,\Real$,
and $u$ has the form
$u \,{=}\, \imath s \,{+}\, u[K]$, $s \,{\in}\, \Real^g$, since $u\,{\in}\, \mathfrak{J}^{\ImN}_K$. 
For brevity we denote $u[K]$ by $\bm{K}$.
Substituting \eqref{WPAllM} into \eqref{SU2Cond},
we find the required identities
\begin{subequations}\label{SU2GammaAlpha}
\begin{align}
&\wp_{1,2N-1}(\imath s \,{+}\, \bm{K}) \wp_{1,2N-1}(-\imath s \,{+}\, \bar{\bm{K}}) = - \rmh_{-1}/ b^2 = \lambda_{4g},
\label{SU2GammaN1}\\
&\frac{\wp_{2i+1,2N-1}(\imath s \,{+}\, \bm{K})}{\wp_{1,2N-1}(\imath s \,{+}\, \bm{K})}
= - \wp_{1,2i-1}({-}\imath s + \bar{\bm{K}}),\ \
i=1,\,\dots,\,N-1, \label{SU2Gamma}\\
&\ImN \frac{ \wp_{1,2i-1,2N-1}(\imath s \,{+}\, \bm{K})}{\wp_{1,2N-1}(\imath s \,{+}\, \bm{K})} = 0,\ \ 
i=1,\,\dots,\,N. \label{SU2Alpha}
\end{align}
\end{subequations}
Here we use the fact that $\overline{\wp_{i,j,\dots}(u)} = \wp_{i,j,\dots}(\bar{u})$,
since  $\wp$ are meromorphic functions generated by \eqref{WPdef} from the $\sigma$-function, which is
an analytic series in $u$ and real parameters $\lambda$ of the curve in question.

Case $\mathfrak{g} \,{=}\, \imath \mathfrak{su}(2)$. Then  $\rmb \,{=}\, b \,{\in}\,\Real$,
and $u$ has the form
$u \,{=}\, s \,{+}\, \bm{K}$, $s \,{\in}\, \Real^g$, since $u\in \mathfrak{J}^{\ReN}_K$. 
Substituting \eqref{WPAllM} into \eqref{iSU2Cond},
we find the required identities in this case
\begin{subequations}\label{iSU2GammaAlpha}
\begin{align}
&\wp_{1,2N-1}(s \,{+}\, \bm{K}) \wp_{1,2N-1}(s \,{+}\, \bar{\bm{K}}) =  \rmh_{-1}/ b^2 = \lambda_{4g},
\label{iSU2GammaN1} \\
&\frac{\wp_{2i+1,2N-1}(s \,{+}\, \bm{K})}{\wp_{1,2N-1}(s \,{+}\, \bm{K})} 
= - \wp_{1,2i-1}(s + \bar{\bm{K}}),\ \
i=1,\,\dots,\,N-1, \label{iSU2Gamma}\\
&\ReN \frac{ \wp_{1,2i-1,2N-1}(s \,{+}\, \bm{K})}{\wp_{1,2N-1}(s \,{+}\, \bm{K})} = 0,\ \ 
i=1,\,\dots,\,N.  \label{iSU2Alpha}
\end{align}
\end{subequations}

Recall, that the Abel pre-image of $\bm{K}$ is composed of $\{e_{2i}\}_{i=1}^N$,
and correspondingly, the Abel pre-image of $\bar{\bm{K}}$ of $\{\bar{e}_{2i}\}_{i=1}^N$. 
It is natural to expect, that the identities \eqref{SU2GammaAlpha}, and \eqref{iSU2GammaAlpha}
hold if $\{\bar{e}_{2i}\}_{i=1}^N$ are among the branch points of $\mathcal{V}$. This means, 
$2g$ branch points form $g$ complex conjugate pairs.

Indeed, there exist $2^g$ even non-singular half-periods,
all located on $\mathfrak{J}^{\ReN}_K$, see  \cite[Theorem\,4]{BerSinG2024}.
In the case of $g$ pairs of complex conjugate branch points, Abel pre-images of these  
$2^g$ half-periods are composed by choosing one point from each pair, see \cite[Proposition\,2]{BerSinG2024}.
Let the half-period $\Omega_{I}$ correspond to the non-special divisor compose of  branch points with 
indices from the set  $I$, $\card I = g$. Thus, $\bm{K} = \Omega_{2,4,\dots,2g}$.
Evidently, \eqref{SU2GammaN1}, and \eqref{iSU2GammaN1} work 
if and only if $\{\bar{e}_{2i}\}_{i=1}^N$ serve as the branch points of $\mathcal{V}$,
and the remaining finite branch point is zero.  
Otherwise, 
\begin{equation}\label{Lambda4g}
\lambda_{4g} = \prod_{i=1}^N e_{2i} \bar{e}_{2i}
\end{equation} 
does not hold. The equality \eqref{Lambda4g} holds if any branch points in $\mathcal{A}^{-1}(\bm{K})$
are replaced by their complex conjugate counterparts. 

The identities \eqref{SU2Gamma}, and \eqref{iSU2Gamma} contain the functions $\wp_{2i+1,2N-1}$,
which are computed from the fundamental relations associated with $\mathcal{V}$,
and expressed through $\wp_{1,2i+1}$, $\wp_{1,1,2i+1}$, $i=1$, \ldots, $g$.

The identities  \eqref{SU2Alpha}, \eqref{iSU2Alpha} can not be verified in this way, since $\wp_{i,j,k}$, and so
the left hand sides, vanish on half-periods.

The same even non-singular half-periods can be arranged on $\mathfrak{J}^{\ImN}_K$.

\smallskip
We prove the identities \eqref{SU2GammaAlpha}, and \eqref{iSU2GammaAlpha}
 in genera $1$, and $2$, using the addition law, see Appendix\;\ref{A:RCAddLaw}. 
 In higher genera,
 the same identities hold in numerical computations for hyperelliptic curves with 
different combinations of real and complex conjugate branch points.
\end{proof}

Note, that only the curves from Conjecture~\ref{C:CCBPs} serve as spectral within the sine-Gordon hierarchy,
 see \cite[Theorem\,5]{BerSinG2024}

\begin{theo}\label{T:MKdVfoc}
The focusing mKdV equation \eqref{MKdVEq}, $\varsigma=-1$, arises on hyperelliptic curves which 
possess a branch point
at infinity, a branch point at the origin, and all other branch points form complex conjugate pairs.  
The real-valued, and bounded finite-gap solution
is given by $w(\rmx,\rmt) \,{=}\, {-} \imath \alpha(\rmx,\rmt)$,
where $\alpha(\rmx,\rmt)$ is defined by \eqref{MKdVSol},  $\rmb \in \Real$, or $\rmb \in \imath \Real$.
\end{theo}
\begin{proof}
The statement follows from Theorem~\ref{T:MKdVESol}, Conjecture~\ref{C:CCBPs},
and Theorem\;\ref{T:Cshift}.
\end{proof}

%-----------------------------------------------
\subsection{Initial conditions}

Every finite-gap solution is represented by a trajectory of a hamiltonian system from the
 hierarchy in question.
An $N$-gap hamiltonian system is fixed by a choice of $\rmr_{-1}$, \ldots, $\rmr_{N-2}$, which fix
an orbit $\mathcal{O}$, and so the phase space.
A trajectory in the phase space is fixed by the values $\rmh_{N-1}$, \ldots, $\rmh_{2N-2}$
of hamiltonians, which serve as initial conditions.

\begin{rem}
In terms of  variables of separation $\{(z_k,w_k)\}_{k=1}^N$, 
an $N$-gap system splits into $N$ one-dimensional systems.
Each system has coordinate $z_k$ and momentum $w_k$, and is governed by 
the hamiltonian
\begin{gather}
\begin{split}
&\mathcal{H}(z_k,w_k) = w_k^2 + \mathcal{U}(z_k),\\
&\mathcal{U}(z_k) = - \big(\rmb^2 z_k^{2N+1} + \rmr_{2N-2} z_k^{2N} + \cdots + \rmr_{N-1} z_k^{N+1} \\
&\phantom{mmmmmmmmmmmmm}  + \rmh_{N-2} z_k^{N} + \cdots + \rmh_{-1} z_k \big).
\end{split}
\end{gather}
Such a system describes the motion of a particle of mass $1/2$ in the potential $\mathcal{U}$.
\end{rem}

With fixed $\rmr_{-1}$, \ldots, $\rmr_{N-2}$, $\rmh_{N-1}$, \ldots, $\rmh_{2N-2}$,
there exists only one trajectory such that $\mathcal{H}(z_k,w_k) = 0$,
which serves as the $N$-gap solution of the extended mKdV equation on a chosen orbit
under chosen initial conditions. This implies that the branch points $(e_i,0)$, $e_i \in \Real$,
serve as turn-points of the trajectory.

%======================================
%-----------------------------------------------
\section{New solutions of the KdV equation}
By applying the Miura transformation  to \eqref{MKdVSol}  we find
\begin{theo} Let
\begin{equation}\label{NewKdVSol}
\gamma(u) =
\rmb^2 \bigg( \lambda_2  - 2 \frac{\wp_{3,2g-1}(u)}{\wp_{1,2g-1}(u)} \bigg)
\end{equation}
Then $\upsilon(\rmx, \rmt) = 
\gamma\big( \rmb (\rmx + 2 \rmr_{2N-2} \rmt,-4 \rmb^2 \rmt, \rmc_5,\dots, \rmc_{2N-1})^t + \bm{C} \big)$
satisfies the KdV equation \eqref{KdVEq}.
\end{theo}
\begin{proof}

Using the identities \eqref{WP1112n1} and
\begin{equation}
\wp_{1,1,2g-1}^2 = 4 \lambda_{4g+2} + 4 \wp_{1,2g-1}^2 \big(\lambda_2
+ \wp_{1,1} \big) - 4  \wp_{1,2g-1}  \wp_{3,2g-1}
\end{equation}
with $\lambda_{4g+2}=0$,
we find that \eqref{NewKdVSol} solves the KdV equation.

$\gamma$ satisfyies
\begin{equation}
-4 \partial_{u_3} \gamma + \partial_{u_1}^3 \gamma + (2 \lambda_2 - 6 \gamma) \partial_{u_1} \gamma = 0.
\end{equation}
\end{proof}

%======================================
\section{Non-linear waves}\label{s:NLW}
In this section we present effective  computation of quasi-periodic 
finite-gap solutions of the sine-Gordon equation.

%------------------------------------------------------------------
\subsection{Numerical computation}
We will use the explicit formula \eqref{SinGSol}.
All computations are made in Wolfram Mathematica 12.
Integrals between branch points are computed with the help of \texttt{NIntegrate}
with \texttt{WorkingPrecision} equal to, or greater than $30$.
The canonical curve  \eqref{V22g1Eq} is defined through its parameters $\lambda_k$,
and roots of the polynomial $\Lambda$ are computed by means of \texttt{NSolve},
with the same \texttt{WorkingPrecision}.

$\wp$-Functions are computed by the formulas
\begin{gather}\label{WPdefComp}
\begin{split}
&\wp_{i,j}(u)  = \varkappa_{i,j} - \frac{\partial^2}{\partial u_i \partial u_j } \theta[K](\omega^{-1} u; \omega^{-1} \omega'),\\
& \wp_{i,j,k}(u)  =- \frac{\partial^3}{\partial u_i \partial u_j \partial u_k} \theta[K](\omega^{-1} u; \omega^{-1} \omega'),
\end{split}
\end{gather}
where $\varkappa_{i,j}$ are entries of the symmetric matrix $\varkappa = \eta \omega^{-1}$, and
 the period matrices $\omega$,  $\omega'$,  $\eta$ are obtained from 
the differentials \eqref{K1DifsGen}, \eqref{K2DifsGen} along the canonical cycles as on
fig.~\ref{cyclesOdd}. Actually, columns of the matrices $\omega$,  $\omega'$,  $\eta$
are computed as follows
\begin{gather*}%\label{K12PerComp}
\begin{split}
  &\omega_k = 2\int_{e_{2k-1}}^{e_{2k}} \rmd u,\qquad\quad \eta_k = 2\int_{e_{2k-1}}^{e_{2k}} \rmd r,\\
  &\omega'_k = - 2 \sum_{i=1}^k \int_{e_{2i-2}}^{e_{2i-1}} \rmd u  
  = 2 \sum_{i=k}^g  \int_{e_{2i}}^{e_{2i+1}} \rmd u.
  \end{split}
\end{gather*}

%=============================

%=================================
\appendix
\section{Reality conditions via the addition law}\label{A:RCAddLaw}

\textbf{Genus 1}. 
Let $\mathcal{V}$ have  branch points at $e_1\,{\in}\,\Real$, $e_2$, $\bar{e}_2$,
and so
$$ \lambda_2 = -e_1-e_2-\bar{e}_2,\qquad 
\lambda_4 = e_1e_2 + e_1 \bar{e}_2 + e_2 \bar{e}_2,\qquad
\lambda_6 = -e_1 e_2 \bar{e}_2. $$

We start with the case of $\mathfrak{g} \,{=}\, \imath \mathfrak{su}(2)$,
In genus $1$ the reality conditions \eqref{iSU2GammaAlpha} acquire the form
\begin{subequations}
\begin{gather}
  \wp(s \,{+}\, \bm{K}) \wp(s \,{+}\, \bar{\bm{K}}) = \lambda_{4}, \label{iSU2GammaG1}\\
\ReN \frac{ \wp'(s \,{+}\, \bm{K})}{\wp(s \,{+}\, \bm{K})} = 0. \label{iSU2AlphaG1}
\end{gather}
\end{subequations}

Applying the addition law to $u \,{=}\, s \,{+}\, \bm{K}$,
and taking into account $\wp(\bm{K}) \,{=}\, e_2$, and $\wp'(\bm{K}) \,{=}\, 0$, we have
\begin{gather}
\wp(s \,{+}\, \bm{K})= \frac{\nu_3^2 - \lambda_6}{e_2 \wp(s)},\qquad
\wp'(s \,{+}\, \bm{K}) = 2 \nu_3 + 2\nu_1 \wp(s \,{+}\, \bm{K}),\\
\nu_1 = \frac{\wp'(s)}{2(\wp(s) - e_2)},\qquad 
\nu_3 = -\frac{e_2 \wp'(s)}{2(\wp(s) - e_2)}.
\end{gather}
Then we find
\begin{equation}
  \wp(s \,{+}\, \bm{K}) \wp(s \,{+}\, \bar{\bm{K}}) 
  =  \frac{ (e_2 \wp(s) - \lambda_4)}{ (\wp(s) - e_2)} 
  \frac{(\bar{e}_2 \wp(s) - \lambda_4)}{(\wp(s) - \bar{e}_2)},
\end{equation}
where the Weierstrass equation is applied
$$\wp'(u)^2 = 4 \big(\wp(u) - e_1\big)\big(\wp(u) - e_2\big)\big(\wp(u) - \bar{e}_2\big),\quad
\forall u\in\Jac(\mathcal{V}) \backslash\{0\}.$$
This expression simplifies if and only if $e_1=0$, which implies  $\lambda_4 = e_2 \bar{e}_2$.
Then, we come to \eqref{iSU2GammaG1}, or more precisely
\begin{gather}\label{iSU2GammaG1M}
  \wp(s \,{+}\, \bm{K}) \wp(s \,{+}\, \bar{\bm{K}}) = e_2 \bar{e}_2. 
\end{gather}

Next, assuming $e_1=0$, and taking into account \eqref{iSU2GammaG1M}, we have
\begin{equation}
\frac{ \wp'(s \,{+}\, \bm{K})}{\wp(s \,{+}\, \bm{K})} 
= \frac{2\nu_3 }{e_2 \bar{e}_2} \wp(s \,{+}\, \bar{\bm{K}}) + 2\nu_1 
= \frac{\wp'(s) (e_2 - \bar{e}_2)}{|\wp(s) - e_2|^2}.
\end{equation}
Evidently, \eqref{iSU2AlphaG1} holds.

The case of $\mathfrak{g} \,{=}\, \mathfrak{su}(2)$ is similar to the above.
We assume that
$\mathcal{V}$ has  branch points at $e_1=0$, $e_2$, $\bar{e}_2$.
The reality consitions \eqref{SU2GammaAlpha} acquire the form
\begin{subequations}\label{SU2GammaAlphaG1}
\begin{gather}
\lambda_{4}  =  \wp(\imath s \,{+}\, \bm{K}) \wp(\imath s \,{+}\, \bar{\bm{K}}),\label{SU2GammaG1} \\
\ImN \frac{ \wp'(\imath s \,{+}\, \bm{K})}{\wp(\imath s \,{+}\, \bm{K})} = 0. \label{SU2AlphaG1}
\end{gather}
\end{subequations}
Applying the addition law to $u \,{=}\, \imath s \,{+}\, \bm{K}$,
 we have
\begin{gather}
\wp(\imath s \,{+}\, \bm{K})= \frac{\nu_3^2}{e_2 \wp(\imath s)},\qquad
\wp'(\imath s \,{+}\, \bm{K}) = 2 \nu_3 + 2\nu_1 \wp(\imath s \,{+}\, \bm{K}),\\
\nu_1 = \frac{\wp'(\imath s)}{2\big(\wp(\imath s) - e_2\big)},\quad 
\nu_3 = -\frac{e_2 \wp'(\imath s)}{2\big(\wp(\imath s) - e_2\big)}
\end{gather}
The same computations with $\imath s$, instead of $s$, prove \eqref{SU2GammaAlphaG1}.

\textbf{Genus 2}. 
We start with an assumption that branch points of $\mathcal{V}$ are $e_1=0$, real $e_2$, $e_3$, 
and complex conjugate $e_4$, $\bar{e}_4$. Then, we show that the reality conditions 
\eqref{SU2GammaAlpha}, and \eqref{iSU2GammaAlpha} are satisfied if and only if $e_3=\bar{e}_2$.

With such an assumption,  $\lambda_{10} = 0$, and  $\lambda_8 = e_2 e_3 e_4 \bar{e}_4$.
Also we take into account that $\wp_{1,1}(\bm{K}) \,{=}\, e_2 + e_4$, $\wp_{1,3}(\bm{K}) \,{=}\, {-} e_2 e_4$, 
and $\wp_{i,j,k}(\bm{K}) \,{=}\, 0$.

The reality conditions \eqref{iSU2GammaAlpha} acquire the form
\begin{subequations}
\begin{gather}
\wp_{1,3}(s \,{+}\, \bm{K}) \wp_{1,3}(s \,{+}\, \bar{\bm{K}}) = \lambda_{8}, \label{iSU2GammaN1G2}\\
\wp_{3,3}(s \,{+}\, \bm{K}) = \wp_{1,3}(s \,{+}\, \bm{K}) \wp_{1,1}(s + \bar{\bm{K}}),\label{iSU2GammaG2} \\
\ReN \frac{ \wp_{1,1,3}(s \,{+}\, \bm{K})}{\wp_{1,3}(s \,{+}\, \bm{K})} = 0,\qquad
\ReN \frac{ \wp_{1,3,3}(s \,{+}\, \bm{K})}{\wp_{1,3}(s \,{+}\, \bm{K})} = 0. \label{iSU2AlphaG2}
\end{gather}
\end{subequations}

Applying the addition law to $u \,{=}\, s \,{+}\, \bm{K}$ we have
\begin{subequations}
\begin{align}
&\wp_{1,3}(s \,{+}\, \bm{K}) = \frac{\nu_6^2}{e_2 e_4 \wp_{1,3}(s)}, \label{WP13AddLaw}\\
&\wp_{1,1}(s \,{+}\, \bm{K})= \frac{1}{e_2 e_4 \wp_{1,3}(s)}\bigg(\lambda_8 \nu_1^2 - 2 \nu_4 \nu_6
+ \nu_6^2 \bigg(\frac{\wp_{1,1}(s)}{\wp_{1,3}(s)} - \frac{e_2 + e_4}{e_2 e_4} \bigg)\bigg),\\
&\wp_{1,1,3}(s \,{+}\, \bm{K}) = 2 \nu_1^{-1} \big(\nu_6 + \nu_2 \wp_{1,3}(s \,{+}\, \bm{K})
+ \wp_{1,1}(s \,{+}\, \bm{K}) \wp_{1,3}(s \,{+}\, \bm{K}) \big),\\
&\wp_{1,1,1}(s \,{+}\, \bm{K}) = 2 \nu_1^{-1} \big(\nu_4 + \nu_2 \wp_{1,1}(s \,{+}\, \bm{K})
+ \wp_{1,1}(s \,{+}\, \bm{K})^2 + \wp_{1,3}(s \,{+}\, \bm{K}) \big),
\intertext{where}
&\nu_6 = e_2 e_4 \Big({-} \wp_{1,3}(s) (\wp_{1,1}(s) - e_2- e_4 )\wp_{1,1,1}(s) \\
&\qquad + \big(\wp_{1,3}(s)+e_2 e_4 + \wp_{1,1}(s) (\wp_{1,1}(s) - e_2- e_4 ) \big) 
\wp_{1,1,3}(s) \Big) \Pi^{-1}, \notag\\ 
&\nu_4 = \Big( \big((e_2+e_4)\wp_{1,3}(s)  \wp_{1,1}(s) - (e_2^2 + e_2 e_4 + e_4^2)\wp_{1,3}(s)
 + e_2^2 e_4^2 \big)\wp_{1,1,1}(s) \\
&\!\!\!\quad + \big((e_2^2 + e_2 e_4 + e_4^2)\wp_{1,1}(s) - (e_2+e_4)(\wp_{1,1}(s)^2 + \wp_{1,3}(s)) \big) 
\wp_{1,1,3}(s) \Big) \Pi^{-1}, \notag \\ 
&\nu_2 = \Big({-} \big(\wp_{1,3}(s) \wp_{1,1}(s)  + e_2 e_4 (e_2+e_4) \big) \wp_{1,1,1}(s)  \\
&\qquad + \big(\wp_{1,3}(s) + \wp_{1,1}(s)^2 - (e_2^2 + e_2 e_4 + e_4^2) \big) \wp_{1,1,3}(s) \Big)\Pi^{-1}, \notag\\ 
&\nu_1 = 2 \big(\wp_{1,3}(s) + e_2 \wp_{1,1}(s) - e_2^2\big)  \big(\wp_{1,3}(s) + e_4 \wp_{1,1}(s) - e_4^2\big) \Pi^{-1},\notag\\
&\Pi = \big(\wp_{1,3}(s)+e_2 e_4\big) \wp_{1,1,1}(s) - (\wp_{1,1}(s) - e_2 - e_4)\wp_{1,1,3}(s). \notag
\end{align}
\end{subequations}
With the help of \eqref{WP13AddLaw} we find
\begin{equation}
\wp_{1,3}(s \,{+}\, \bm{K}) \wp_{1,3}(s \,{+}\, \bar{\bm{K}}) = 
\frac{\nu_6^2 \bar{\nu}_6^2}{e_2 \bar{e}_2 e_4 \bar{e}_4 \wp_{1,3}(s)^2},
\end{equation}
and
\begin{multline}\label{Nu6Nu6Bar}
\frac{\nu_6 \bar{\nu}_6}{e_2 \bar{e}_2 e_4 \bar{e}_4} 
= - \wp_{1,3}(s) + \Big( \big[2 \wp_{1,3}(s) \Phi_8 - \wp_{1,1}(s) \Phi_{10}\big]  \wp_{1,1,1}(s)^2 \\
- \big[ 2 \wp_{1,3}(s) \Phi_6 + 2 \Phi_{10} \big] \wp_{1,1,1}(s)\wp_{1,1,3}(s)
+ \big[\wp_{1,1}(s) \Phi_6 + \Phi_8\big]\Big) (\Pi \bar{\Pi})^{-1}  
\end{multline}
where
\begin{multline*}
\Pi \bar{\Pi} =
  \Big(\big[\wp_{1,3}(s)^2 + (e_2 e_4 + \bar{e}_2 \bar{e}_4) \wp_{1,3}(s) + \lambda_8 \big] \wp_{1,1,1}(s)^2 \\
 - \big[ \wp_{1,3}(s) (2 \wp_{1,1}(s) + \lambda_2) +  (e_2 e_4 + \bar{e}_2 \bar{e}_4) \wp_{1,1}(s)
 + \lambda_6 \big] \wp_{1,1,1}(s)\wp_{1,1,3}(s) \\
 + \big[\wp_{1,1}(s)^2 + \lambda_2 \wp_{1,1}(s) + (e_2 + e_4)(\bar{e}_2 + \bar{e}_4) \big] \wp_{1,1,3}(s)^2 \Big)^{-1},
\end{multline*}
and the fundamental identities in genus $2$ are employed
\begin{gather}
\begin{split}
&0 = \Phi_6 \equiv \wp_{1,1}(u)^3  + \lambda_2 \wp_{1,1}(u)^2  + \wp_{1,1}(u)\wp_{1,3}(u) 
+ \lambda_4  \wp_{1,1}(u) + \lambda_6 \\
&\phantom{mmmmmmmmmmmmmmmmmmmm} +  \wp_{3,3}(u) - \tfrac{1}{4} \wp_{1,1,1}(u)^2,\\
&0 = \Phi_8 \equiv \wp_{1,1}(u) \wp_{1,3}(u) \big(\wp_{1,1}(u) + \lambda_2\big)
-  \tfrac{1}{4} \wp_{1,1,1}(u) \wp_{1,1,3}(u) \\
&\qquad\qquad\quad\ 
+ \tfrac{1}{2} \big(\wp_{1,3}(u)^2 - \wp_{1,1}(u) \wp_{3,3}(u) + \lambda_4 \wp_{1,3}(u) + \lambda_8 \big), \\
&0 = \Phi_{10}  \equiv \wp_{1,3}(u)^2 \big(\wp_{1,1}(u) + \lambda_2\big) - \wp_{1,3}(u) \wp_{3,3}(u)
- \tfrac{1}{4} \wp_{1,1,3}(u)^2.
\end{split}
\end{gather}
Thus, the second term on the right hand side of \eqref{Nu6Nu6Bar} vanishes,
and we obtain
\begin{equation}
\frac{\nu_6 \bar{\nu}_6}{e_2 \bar{e}_2 e_4 \bar{e}_4}  = - \wp_{1,3}(s),
\end{equation}
which implies \eqref{iSU2GammaN1G2}.
Note, such a simplification takes place if and only if four branch points form complex conjugate pairs, and the
remaining finite branch point is zero.

Next, we consider \eqref{iSU2GammaG2}, where 
\begin{multline}
 \wp_{1,3}(s \,{+}\, \bm{K}) \wp_{1,1}(s + \bar{\bm{K}}) \\
=  \frac{\nu_6^2}{e_2 e_4 \bar{e}_2 \bar{e}_4\wp_{1,3}(s)^2}
\bigg(\lambda_8 \bar{\nu}_1^2 - 2 \bar{\nu}_4 \bar{\nu}_6
+ \bar{\nu}_6^2 \bigg(\frac{\wp_{1,1}(s)}{\wp_{1,3}(s)} - \frac{\bar{e}_2 + \bar{e}_4}{\bar{e}_2 \bar{e}_4} \bigg)\bigg) \\
=  \frac{\nu_6^2 \big(\lambda_8 \bar{\nu}_1^2 - 2 \bar{\nu}_4 \bar{\nu}_6 \big)}
{e_2 e_4 \bar{e}_2 \bar{e}_4\wp_{1,3}(s)^2} +
\lambda_8
\bigg(\frac{\wp_{1,1}(s)}{\wp_{1,3}(s)} - \frac{\bar{e}_2 + \bar{e}_4}{\bar{e}_2 \bar{e}_4} \bigg)
 \end{multline}
 Then, $\wp_{3,3}(u)$ is obtained from the identity $\Phi_6 = 0$
\begin{multline}
\wp_{3,3}(s \,{+}\, \bm{K}) = \frac{1}{\nu_1^2} \Big(\nu_4 + \nu_2 \wp_{1,1}(s \,{+}\, \bm{K})
+ \wp_{1,1}(s \,{+}\, \bm{K})^2 + \wp_{1,3}(s \,{+}\, \bm{K}) \Big)^2 \\
 - \big(\wp_{1,1}(s \,{+}\, \bm{K})^3  + \lambda_2 \wp_{1,1}(s \,{+}\, \bm{K})^2  
 + \wp_{1,1}(s \,{+}\, \bm{K})\wp_{1,3}(s \,{+}\, \bm{K}) \\
+ \lambda_4  \wp_{1,1}(s \,{+}\, \bm{K}) + \lambda_6 \big).
\end{multline} 

The first expression in \eqref{iSU2AlphaG2} has the form
\begin{multline}
 \frac{ \wp_{1,1,3}(s \,{+}\, \bm{K})}{\wp_{1,3}(s \,{+}\, \bm{K})}
  = \frac{2 \nu_6}{\nu_1 \wp_{1,3}(s \,{+}\, \bm{K})} 
+ \frac{2}{\nu_1} \big(\nu_2 + \wp_{1,1}(s \,{+}\, \bm{K}) \big) \\
= \frac{2 \nu_6}{\nu_1 \lambda_8} \wp_{1,3}(s \,{+}\, \bar{\bm{K}}) 
 + \frac{2 \nu_2}{\nu_1} +  \frac{2}{\nu_1} \wp_{1,1}(s \,{+}\, \bm{K}) \\
 =  \frac{2 \nu_6}{\nu_1 \lambda_8} \frac{\bar{\nu}_6^2}{\bar{e}_2 \bar{e}_4 \wp_{1,3}(s)}
  + \frac{2 \nu_2}{\nu_1} +  \frac{2}{\nu_1} 
  \frac{1}{e_2 e_4 \wp_{1,3}(s)}\bigg(\lambda_8 \nu_1^2 - 2 \nu_4 \nu_6
+ \nu_6^2 \bigg(\frac{\wp_{1,1}(s)}{\wp_{1,3}(s)} - \frac{e_2 + e_4}{e_2 e_4} \bigg)\bigg) \\
= \frac{2 e_2 e_4 \wp_{1,3}(s)}{\nu_1 \nu_6} + \frac{2 \nu_2}{\nu_1} 
+  \frac{2 \lambda_8 \nu_1}{e_2 e_4 \wp_{1,3}(s)}
- \frac{4 \nu_4 \nu_6}{e_2 e_4 \nu_1  \wp_{1,3}(s)}
+ \frac{2 \nu_6^2}{e_2 e_4 \nu_1 \wp_{1,3}(s)} \bigg(\frac{\wp_{1,1}(s)}{\wp_{1,3}(s)} - \frac{e_2 + e_4}{e_2 e_4} \bigg).
 \end{multline}

\end{document}